\documentclass[12pt]{report}
           \usepackage{latexsym,amssymb,amsmath}
           \usepackage{amsthm}
           \usepackage{arydshln}
           \usepackage{graphics}
           \usepackage[latin1]{inputenc}
           \usepackage{graphicx}
           \usepackage{setspace}
\DeclareMathOperator{\tr}{tr}
\DeclareMathOperator{\sign}{sign}
\DeclareMathOperator{\p}{p}
\DeclareMathOperator{\E}{E}
\DeclareMathOperator{\diag}{diag}

\newcommand{\ket}[1]{\left| #1 \right\rangle}
\newcommand{\bra}[1]{\left\langle #1 \right|}
\newcommand{\ketbra}[2]{\left| #1 \right\rangle \left\langle #2 \right|}
\newcommand{\ketbraa}[1]{\left| #1 \right\rangle \left\langle #1 \right|}
\newcommand{\kett}[2]{\left| #1 #2 \right\rangle}
\newcommand{\braa}[2]{\left\langle #1 #2\right|}

\newtheorem{coro}{Corollary}
\newtheorem{notacao}{Notation}
\newtheorem{prop}{Proposition}[section]
\newtheorem{exemplo}{Example}
\newtheorem{definicao}{Definition}[chapter]

\newtheorem{teo}{Theorem}
\newtheorem{lema}{Lemma}

\newtheorem{proto}{Protocol}

\begin{document}

\onehalfspacing

\thispagestyle{empty}
\begin{titlepage}

\begin{center}
{\Large
Universidade Federal do Rio Grande do Sul \\[2mm]
Instituto de Matem\'atica}
\end{center}

\vfill

\begin{center}

\renewcommand{\thefootnote}{\fnsymbol{footnote}}
\setcounter{footnote}{0}
{ \LARGE {\bf Local models and hidden nonlocality in Quantum Theory} }
\vspace{1cm} \\
{ \Large
Leonardo Guerini de Souza}

\end{center}

\vfill

\hfill
\parbox{.55\textwidth}{{\sl Master's thesis written under the supervision of Alexandre Tavares Baraviera,
presented to Programa de P\'os-Gradua\c{c}\~ao em Matem\'atica - UFRGS
as a partial requirement for the title of Master in Mathematics.} \/ }
\vfill
\begin{center}
{\it Porto Alegre, March of 2014. \/}
\end{center}
\end{titlepage}
\renewcommand{\thefootnote}{\arabic{footnote}}

\thispagestyle{empty}
\par Master's thesis submitted by Leonardo Guerini de Souza\footnote{Supported by Conselho Nacional de Desenvolvimento Cient\'ifico e Tecnol\'ogico (CNPq).}\footnote{E-mail adress: guerini.leonardo@gmail.com} to Programa de P\'os-Gradua\c{c}\~ao em Matem\'atica of Universidade Federal do Rio Grande do Sul, as a partial requirement for the title of Master in Mathematics.
\vspace{16pt}
\par Date: April 10, 2014.
\vspace{16pt}
\par Supervisor:

Dr. Alexandre Tavares Baraviera

\vspace{16pt}
\par Examiners:

Dr. Carlos Felipe Lardizabal Rodrigues (IM-UFRGS)

Dr. Marcelo de Oliveira Terra Cunha (DMAT-UFMG)

Dra. Sandra Denise Prado (IF-UFRGS)

\thispagestyle{empty}
\noindent
{\Large \bf Abstract}

\indent

\noindent

This Master's thesis has two central subjects: the simulation of correlations generated by local measurements on entangled quantum states by local hidden-variables models and the revelation of hidden nonlocality. We present and detail the Werner's local model and the hidden nonlocality of some Werner states of dimension $d\geq5$, the Gisin-Degorre's local model for a Werner state of dimension $d=2$ and the local model of Hirsch \emph{et al.} for mixtures of the singlet state and noise, all of them for projective measurements. Finally, we introduce the local model for POVMs of Hirsch \emph{et al.} for a state constructed upon the singlet with noise, that still violates the CHSH inequality after local filters are applied, hence presenting the so-called genuine hidden nonlocality.

\setcounter{page}{1}
\pagestyle{plain}

\tableofcontents

\newpage

\addcontentsline{toc}{chapter}{Introduction}
\chapter*{Introduction}

Quantum Theory is an intrinsically probabilistic theory, that is, when we perform a measurement on a quantum system, we can only tell the probabilities associated to each possible outcome, and not the outcome itself. In the situation where the system is composed by two distinct parties, we can perform a local measurement on each party, generating a distribution of joint probabilities, referent to each outcome obtained. Generally, joint probabilities (originated by quantum measurements or not) are not independent, i.e., cannot by expressed by the product of the probability of each party. However, considering the existence of an extra information, which we call local hidden variables, in certain cases such factorization becomes possible. When we are restricted to correlations generated by measurements on quantum states, the previous sentence can be rewritten as the following: some quantum states admit a local hidden variables model, or, in short, a local model. Those are said to be local states. The construction of a local model is a hard task, and even 25 years after the first of them, created by Reinhard Werner \cite{Wer}, the number of states for which were constructed a local model is still small.

The common characteristic shared by correlations that do not admit a local factorization even when local hidden variables are considered is called nonlocality. Another way to characterize this property is through Bell inequalities \cite{Bell64}. A Bell inequality is a relation that should be satisfied by all local correlations, particularly by all correlations generated by measurements on local states. While the absence of a local model for a given state does not imply its nonlocality, the violation of a single Bell inequality is enough to attest it.

Separable quantum states are naturally local. Maybe just as interesting as the use of local hidden variables to simulate correlations of entangled quantum states is the fact that some local states do violate a Bell inequality when subjected to a sequence of local measurements \cite{Pop}. In other words, even if there is a local model for a given state, maybe there are local measurements such that one of the possible after-measurement states violates a Bell inequality. Such intermediate measurements are called filters, and a state that allows this possibility is said to present hidden nonlocality. Therefore, the action of the filters is to reveal the state's nonlocality.

The main objective of this dissertation is to present some examples of local models and some cases of hidden nonlocality. In Chapters 1 and 2 are presented definitions and basic results from the topics around the subject, as well as a short overview of Quantum Theory. This is the starting point to discuss the EPR experiment \cite{EPR}, which motivated the appearance of Bell inequalities.

In Chapter 3 we present the CHSH inequality \cite{CHSH}, which is the only Bell inequality employed during the thesis, as well as the Horodecki Criterion, which gives us the largest CHSH violation that a given state provides \cite{HHH}. We also formally define the role of the local hidden variables \cite{BCPSW} and the construction of a local model.

Chapter 4 is turned to Werner's local model for projective measurements \cite{Wer}, the first local model to arise, which influenced basically all the subsequent works, being one of the most important papers in the area\footnote{It was in this paper where first appears the definition of entanglement for mixed states, for example, property until then defined only for pure states.}. We motivate the study of Werner states, which are the ones used for its construction, and present the hidden variables and response functions to be used in each system.

In Chapter 5 we present the work of Sandu Popescu \cite{Pop}, showing that some of the states for which we constructed the Werner's local model present hidden nonlocality: applying a simple choice of local filters, we are able to obtain a state that violates maximally the CHSH inequality, when the local dimension is $d\geq5$.

Chapter 6 brings the Gisin-Degorre's local model \cite{GG,DLR}, capable of simulating the EPR experiment if we allow classical communication between the parties. However, even without communication or any other additional resource, we are able to simulate a Werner state of local dimension $d=2$, for projective measurements.

Finally, Chapter 7 is based on the work of Flavien Hirsch, Marco T\'ulio Quintino, Joseph Bowles and Nicolas Brunner \cite{HQBB}. We start by showing the existence of hidden nonlocality in a state of local dimension $d=2$ and then present an example of the so-called genuine hidden nonlocality: we provide a state which has a local model for POVMs and, nevertheless, violates the CHSH inequality after the application of local filters.

\chapter{Preliminaries}

In this chapter we present the standard notation while dealing with Quantum Theory, introduced by Paul Dirac, and a miscellaneous list of basic results and definitions that will be used throughout the thesis.

\section{Dirac's notation}

Dirac's notation is a mnemonic notation that is very useful in the handling with Quantum Theory. The elements of a vector space $V$ are denoted by $\ket v$ and the elements of the dual $V^*$ are denoted by $\bra v$. The symbols $\ket{\cdot}$ and $\bra{\cdot}$ already defines the object; the letters only serve as labels. Consequently, the canonical inner product of vectors $u$ and $v$ is simply written as $\bra u \ket v\equiv \bra u v\rangle$.

Along the thesis, sometimes we will want to speak in the adjoint operator $T^{\dagger}$ of a given operator $T$ (the proper definitions will be given in the next section). While in the traditional notation we write
\begin{equation*}
  \langle v, Tu\rangle= \langle T^{\dagger}v,u \rangle,
\end{equation*}
in Dirac's notation both sides of the above expression are written the same,
\begin{equation*}
  \bra v (T \ket u)=(\bra v T)\ket u=\bra v T \ket u.
\end{equation*}
This happens because the dual element of $T^{\dagger}\ket v$ is exactly $\bra vT$. Sometimes, to emphasizing this step, we will make use of the abuse of notation
\begin{equation*}
  \bra v Tu\rangle:=\bra v T \ket u=:\bra{T^{\dagger}v}u\rangle.
\end{equation*}

\section{Basic definitions and results}

\subsection{Linear algebra}

\begin{definicao} A Hilbert space is a pair $(V, \langle\cdot,\cdot\rangle)$, where $V$ is a vector space and $\langle\cdot,\cdot\rangle$ is an inner product that induces a distance function for which $V$ is complete, that is, every Cauchy sequence is convergent.

Except if mentioned otherwise, the inner product that we will be considering is
\begin{equation*}
  \langle\ket v ,\ket w \rangle = \bra v | w\rangle.
\end{equation*}
\end{definicao}

\begin{definicao} Let $V,W$ be vector spaces over the same field. The tensor product of them, denoted by $V\otimes W$, is the space generated by vectors of the form $\ket v \otimes \ket w$, where $\ket v \in V$ and $\ket w \in W$, which obey to the following relations:
\begin{enumerate}
  \item $(\lambda\ket v)\otimes \ket w=\ket v \otimes (\lambda \ket w)=\lambda(\ket v \otimes \ket w)$;
  \item $(\ket v +\ket {v'})\otimes\ket w= \ket v \otimes \ket w + \ket {v'}\otimes \ket w$;
  \item $\ket v \otimes (\ket w + \ket {w'})=\ket v \otimes \ket w + \ket v \otimes \ket {w'}$;
  \item If $(V, \langle\cdot,\cdot\rangle_V)$ and $(W, \langle\cdot,\cdot\rangle_W)$ are Hilbert spaces, then $(V\otimes W, \langle\cdot,\cdot\rangle)$ is a Hilbert space, where
      \begin{equation*}
        \langle (\ket v\otimes \ket{w}),(\ket {v'} \otimes \ket{w'})\rangle_{V\otimes W}=\bra v |v'\rangle \bra {w} |w'\rangle,
      \end{equation*}
      with $\ket{v'}\in V, \ket{w'}\in W$.
\end{enumerate}
\end{definicao}

\begin{notacao} We will use the simplified notation $\ket v \otimes \ket w \equiv \ket v \ket w \equiv \kett vw $ for vectors of a tensor product of two spaces.
\end{notacao}

\begin{definicao} Let $A\in\mathcal{L(H)}=\{T:\mathcal{H}\rightarrow \mathcal{H}; \ T\ is\ linear\}$, where $\mathcal{H}$ is a finite dimensional Hilbert space, and $[A]=(a_{ij})$ be a matrix representation of $A$. The trace of $A$ is the sum of the elements of the main diagonal of $[A]$, that is,
\begin{equation*}
\tr(A)=\sum_i{a_{ii}}.
\end{equation*}
\end{definicao}

\begin{prop}\label{kamus} Given a d-dimensional Hilbert space $\mathcal{H}$, let $A,B \in \mathcal{L(H)}$ and $V:\mathcal{H}\otimes \mathcal{H}\rightarrow\mathcal{H}\otimes \mathcal{H}$ be the ``flip" linear operator, defined by $V(\ket {ab})=\ket {ba}$. Then $\tr(V A\otimes B)=\tr(AB)$.
\end{prop}

\begin{proof} Fix an orthonormal base $\mathcal{B}=\{\ket1, ... , \ket d\}$ for $\mathcal{H}$. Note that we can write the matrix representation $A=(a_{ij})_{i,j=1}^d$ as
\begin{equation*}
A=\sum_{i,j=1}^d{a_{ij}\ketbra{i}{j}}.
\end{equation*}
Indeed, $\ketbra{i}{j}$ is the linear operator that vanishes for every vector of $\mathcal{B}$ except $\ket j$, hence all the entries of its matrix representation are null except for the the one in the $j$-th column and $i$-th line. Similarly for $B$,
\begin{equation*}
B=\sum_{k,l=1}^d{b_{kl}\ketbra{k}{l}}.
\end{equation*}

Therefore,
\begin{eqnarray*}
  \tr(AB) &=& \tr\left(\left[\sum_{i,j=1}^d{a_{ij}\ketbra{i}{j}}\right]\left[\sum_{k,l=1}^d{b_{kl} \ketbra{k}{l}}\right]\right) \\
   &=& \tr\left(\sum_{i,j,l}^d{a_{ij}b_{jl}\ketbra{i}{l}}\right) \\
   &=& \sum_{i,j,l}^d{a_{ij}b_{jl}\tr\left(\ketbra{i}{l}\right)} \\
   &=& \sum_{i,j}^d{a_{ij}b_{ji}}.
\end{eqnarray*}

In the same line, the tensor product becomes
\begin{equation}\label{miro}
  A\otimes B=\sum_{i,j,k,l}{a_{ij}b_{kl}\ketbra{ki}{lj}}.
\end{equation}

As for the flip operator, since it can be seen as a permutation over the elements of the base of $\mathcal{H}\otimes \mathcal{H}$, we can describe the action of its matrix representation on another matrix as a simply permutation of the matrix lines. We can express this fact as
\begin{equation}\label{mu}
V \cdot \ketbra{ij}{kl}=(V\ket{ij})\bra{kl}=\ketbra{ji}{kl}.
\end{equation}

Now we are ready to compute $\tr( V A\otimes B)$. Using Eqs. (\ref{miro}) and (\ref{mu}), we have
\begin{eqnarray*}
  \tr(V A\otimes B) &=& \tr\left(\sum_{i,j,k,l}{a_{ij}b_{kl}V(\ketbra{ki}{lj})}\right) \\
   &=& \sum_{i,j,k,l}{a_{ij}b_{kl}\tr(\ketbra{ik}{lj})} \\
   &=& \sum_{i,j}{a_{ij}b_{ji}},
\end{eqnarray*}
thus equal to $\tr(AB)$.
\end{proof}

\begin{notacao} We will denote a $n \times n$ diagonal matrix $A=(a_{ij})$ (i.e., a matrix $A$ for which $a_{ij}=0$ for all $i\neq j$) by $A= \diag(a_{11},...,a_{nn})$.
\end{notacao}

\begin{prop}\label{probpureprojec} Given a Hilbert space $\mathcal{H}$, an unit vector $\ket{\psi}\in \mathcal{H}$ and a linear operator $A$ acting on $\mathcal{H}$, we have
\begin{equation*}
  \tr(A\ket{\psi}\bra{\psi})=\bra{\psi}A\ket{\psi}.
\end{equation*}
\end{prop}
\begin{proof} Consider the orthogonal basis $\{\ket{e_i}\}$ of $\mathcal{H}$, where $\ket{e_1}=\ket{\psi}$, and the matrix representation $[A]$ of $A$ in this basis. The first column of $[A]$ is the vector $A\ket{\psi}=A\ket{e_1}= \sum_i{a_i\ket{e_i}}$. Thus
\begin{equation*}
  \bra{\psi}A\ket{\psi}=\bra{e_1}A\ket{e_1}=\bra{e_1}(\sum_i{a_i\ket{e_i}})=a_1\bra{e_1}\ket{e_1}=a_1.
\end{equation*}

On the other hand, the matrix representation of the projector $\ket{\psi}\bra{\psi}$ in this basis is $\diag(1,0,...)$, so the product $A\ket{\psi}\bra{\psi}$ equals to the matrix which the first column is the vector $A\ket{\psi}$ and the rest of matrix elements equals to zero, resulting that $\tr(A\ket{\psi}\bra{\psi})=a_1$.
\end{proof}

\begin{definicao} An operator $T\in\mathcal{L(H_A\otimes H_B)}$ is said to be a product operator if $T\in\mathcal{L(H_A)\otimes L(H_B)}=\{T_A\otimes T_B\in\mathcal{L(H_A\otimes H_B)}; T_A\in\mathcal{L(H_A)}, T_B\in\mathcal{L(H_B)}\}$.
\end{definicao}

\begin{prop}\label{siegfrid} Every non-product operator can be written as a linear combination of product operators.
\end{prop}

\begin{proof} The proposition follows from the fact that if $\{\ket{a_i}\}$ is a basis for $\mathcal{H}_A$ and $\{\ket{b_j}\}$ is a basis for $\mathcal{H}_B$, then $\{\ketbra{a_ib_j}{a_{i'}b_{j'}}\}=\{\ketbra{a_i}{a_{i'}}\otimes\ketbra{b_j}{b_{j'}}\}$ is a basis for $\mathcal{L(H_A\otimes H_B)}$.
\end{proof}

Since $\mathcal{L(H)}$ is also a vector space, the tensor product $\mathcal{L(H_A)}\otimes \mathcal{L(H_B)}$ is well defined. We can define the action of $T_A\otimes T_B: \mathcal{L(H_A \otimes H_B})\rightarrow \mathcal{L}(\mathcal{H_A \otimes H_B)}$ by
\begin{equation*}
  T_A \otimes T_B (M_A\otimes M_B)=(T_AM_A)\otimes(T_BM_B)
\end{equation*}
for product operators and extend to non-product operators by linearity.

\begin{definicao} Let $A\in\mathcal{L(H)}$ and $[A]$ the matrix representation of $A$ with respect to the base $\mathcal{B}$. The adjoint of $A$ is the operator $A^{\dagger}\in\mathcal{L(H)}$ whose matrix representation with respect to $\mathcal{B}$ is $[A]$ transposed and complex conjugated.
\end{definicao}

\begin{definicao} Let $T\in\mathcal{L(H)}$. $T$ is said to be
\begin{itemize}
  \item Hermitian if $T=T^{\dagger}$;
  \item normal if $TT^{\dagger}=T^{\dagger}T$;
  \item positive semi-definite if $\bra{\psi}T\ket{\psi}\geq0, \ \forall\ket{\psi}\in\mathcal{H}$;
  \item unitary if $T^{\dagger}=T^{-1}$.
\end{itemize}

\end{definicao}

The next two basic results about these classes of operators we will only enounce; one can find the proofs in \cite{elon}.

\begin{teo}
\begin{description}
  \item[(i)] Any positive semi-definite operator on a vector space is a Hermitian operator.
  \item[(ii)] Any Hermitian operator on a vector space is a normal operator.
  \item[(iii)] Any unitary operator on a vector space is a normal operator.
\end{description}
\end{teo}

\begin{teo}\label{spectral} [Spectral Decomposition Theorem] An operator on a vector space is normal if and only if is diagonal with respect to some basis for the space.
\end{teo}

\begin{prop} Let $T,P\in\mathcal{L(H)}$ such that $P=\ketbraa{v}$ is a projector. Then $PT=TP$ if and only if $\ket v$ is an eigenvector of $T$.
\end{prop}
\begin{proof} Suppose that $TP=PT$. Then
\begin{equation*}
T\ket v=TP\ket v=PT\ket v=\lambda\ket v,
\end{equation*}
for some $\lambda\in\mathbb{C}$, by definition of $P$.

Conversely, suppose $T\ket v=\lambda \ket v$ and take $\ket u\in\mathcal{H}$. Consider the base $\{\ket{v_i}\}$ of $\mathcal{H}$ such that $\ket v=\ket{v_1}$. Then $\ket u=\sum_i{a_i\ket{v_i}}$, with $a_i\in\mathbb{C}$, and
\begin{equation*}
TP\ket u=T(a_1\ket{v})=a_1T\ket v=a_1\lambda\ket{v}.
\end{equation*}
On the other hand,
\begin{equation*}
PT\ket u=P(a_1 \lambda \ket v + \sum_{i>1}{a_1T\ket{v_i}})=a_1\lambda\ket v,
\end{equation*}
completing the proof.
\end{proof}

\begin{coro}\label{tchatchatcha} Let $S,T\in\mathcal{L(H)}$ such that $S$ is Hermitian and its spectral decomposition is $\sum_i{a_iP_i}$. If $TP_i=P_iT, \ \forall i$, then $T=\sum_i{b_iP_i}$.
\end{coro}

\subsection{Basic Probability Theory}

\begin{definicao} Let $A,B$ be random variables. The conditional probability that $B=b$ given that $A=a$ is defined by
\begin{equation*}
  \p(B=b|A=a)=\frac{\p(A=a,B=b)}{\p(A=a)}.
\end{equation*}
When $\p(A=a)=0$ we make the convention that $\p(B=b|A=a)=0$.
\end{definicao}

\begin{definicao} Random variables $A,B$ are said to be independent if $\p(A=a, B=b)=\p(A=a)\p(B=b)$.
\end{definicao}

\begin{notacao} We often denote $\p(A=a, B=b)$ by $\p(a,b)$, leaving the ``A=" and ``B=" implicit.
\end{notacao}

\begin{teo}[Law of total probability] If $A,B$ are random variables, then
\begin{equation*}
  \p(b)=\sum_a{\p(b|a)\p(a)},
\end{equation*}
where the sum is over all values $a$ that A can assume.
\end{teo}

\begin{definicao} The expectation or expected value of a random variable $A$ that take values in $\mathbb{R}$ is defined by
\begin{equation*}
  E(A)=\sum_a{\p(a)a}
\end{equation*}
where the sum is over all values $a$ that A can assume.
\end{definicao}

\begin{prop} The expectation has the following properties.
\begin{description}
  \item[(i)] $E(A)$ is linear in $A$.
  \item[(ii)] If $A,B$ are independent, then $E(AB)=E(A)E(B)$.
\end{description}
\end{prop}

We recommend \cite{barry} for further definitions and results.

\chapter{Quantum Theory}

This chapter has no intention to be a didactic introduction to Quantum Theory and should not be the first text about the subject to be readed. Its purpose is to present the minimum of the mathematical framework related to the quantum operations and phenomena that we are interested in this text, such as measurements and nonlocality. Therefore, no physical motivation will be exposed and various of basic and important topics will be completely ignored (such as time evolution of quantum systems, for example). For this reason, the postulates of Quantum Theory will be exposed as definitions. On the other hand, some very specific results will have to appear.

We will start directly making use of the density operator formalism. For an introduction to Quantum Theory and its formalisms we suggest references \cite{NC}, \cite{ABC} and \cite{TC}.

\section{States}

In Quantum Theory, we postulate that a system is associated to a Hilbert space $\mathcal{H}$.\footnote{Since throughout this text the Hilbert spaces will have finite dimension, we can think that each system we are going to speak about is associated to $\mathbb{C}^d$, for some dimension $d$.} Considering the set $\mathcal{L(H)}$ of linear operators of $\mathcal{H}$, a state is an element of $\mathcal{L(H)}$ that describes completely the system.

\begin{definicao} A state of a system associated to $\mathcal{H}$ is an operator $\rho\in \mathcal{L(H)}$ which is positive semi-definite with unit trace.
\end{definicao}

The subset of $\mathcal{L(H)}$ formed by states is denoted by $\mathcal{D(H)}$, where the ``$\mathcal{D}$" comes from the expression ``density operator", another term for ``quantum state"  in this formalism. $\mathcal{D(H)}$ is a convex set, that is, every convex combination of density matrices is also a density matrix. If the state $\rho$ is a one-dimensional projector (that is, if $\rho^2=\rho$ and $Im (\rho)$ has dimension 1), then $\rho=\ket{\psi}\bra{\psi}$ for some $\ket{\psi} \in \mathcal{H}$, thus we can identify the density operator $\rho$ with the vector $\ket{\psi}$. The unit trace condition implies that $\ket {\psi}$ has euclidean norm equals to 1. Every other kind of density operator is called a \emph{mixed state} and can be written as a convex combination of projectors, i.e.,
\begin{equation}\label{charlie}
  \rho \in \mathcal{D(H)} \implies \rho=\sum_i{\p_i\ket{\psi_i}\bra{\psi_i}},
\end{equation}
with unit $\ket{\psi_i}\in \mathcal{H}$ and $\p_i\geq0$ satisfying $\sum_i{\p_i}=1$.
Notice that a one-dimensional projector is a mixed state with only one term on the sum. Therefore, it is called a \emph{pure state}. However, the decomposition in (\ref{charlie}) is not unique. Take for example

\begin{exemplo}
\begin{equation*}
\rho=\frac{3}{4}\ket{0}\bra{0}+\frac{1}{4}\ket{1}\bra{1}=\frac{1}{2}\ket{a}\bra{a}+\frac{1}{2}\ket{b}\bra{b},
\end{equation*}
where $\ket{a}=\frac{\sqrt3}{2}\ket0+\frac{1}{2}\ket1$ and $\ket{b}= \frac{\sqrt3}{2}\ket0-\frac{1}{2}\ket1$.
\end{exemplo}

In this thesis, we will turn our attention mostly to composite systems; generally, we will be considering two different systems $\mathcal{H_A}$ and $\mathcal{H_B}$. The right way to describe such compositions is through the tensor product of Hilbert spaces, which is a Hilbert space itself.

\begin{definicao} The state space of a composite system is the tensor product of the state spaces of the component systems.
\end{definicao}

In the product $\mathcal{H}_A\otimes \mathcal{H}_B$, the first factor will be said to be Alice's system and the second, Bob's system, following the usual terminology used in the literature.

In a composite system, we can speak about separability.

\begin{definicao} A pure state $\ket{\psi}\in \mathcal{H}_A\otimes \mathcal{H}_B$ is said to be a product state if there are $\ket{\psi_A} \in \mathcal{H}_A$ and $\ket{\psi_B} \in \mathcal{H}_B$ such that
\begin{equation}\label{separable_pure}
  \ket{\psi}=\ket{\psi_A}\otimes \ket{\psi_B}.
\end{equation}
The pure states that are not product are said to be entangled.
\end{definicao}

The concept of product state is just the restriction of the term product operator to the set $\mathcal{D(H)}$. However, the analogy does not hold for Proposition \ref{siegfrid}, since it is false that every state is a convex combination of product states.

As any definition given by a denial, it is usually hard to decide whether a state is entangled or not: in principle, one needs to show that equation (\ref{separable_pure}) does not hold for any pair $\ket{\psi_A}\otimes\ket{\psi_B} \in \mathcal{H}_A\otimes\mathcal{H}_B$.

\begin{exemplo} Let $\mathcal{H}=\mathbb{C}^2 \otimes \mathbb{C}^2$ and $\mathcal{B}=\{\ket0, \ket1\}$ be a basis of $\mathbb{C}^2$. The state
\begin{equation*}
  \ket{\psi}=\frac{1}{\sqrt2}(\ket{00}+\ket{01})
\end{equation*}
of $\mathcal{H}$ is a product state, since
\begin{equation*}
  \ket{\psi}=\ket{0}\otimes\frac{1}{\sqrt2}(\ket{0}+\ket{1}).
\end{equation*}
\end{exemplo}

\begin{exemplo} Let $\mathcal{H}$ and $\mathcal{B}$ be the same as above but take
\begin{equation*}
  \ket{\Phi_+}=\frac{1}{\sqrt2}(\ket{00}+\ket{11}).
\end{equation*}
If $\ket{\Phi_+}$ were separable, then there would be $\ket{a},\ket{b} \in \mathbb{C}^2$ such that
\begin{equation}\label{bla}
\ket{\Phi_+}=\ket{a}\otimes\ket{b}.
\end{equation}
Writing these vectors in terms of $\mathcal{B}$, we have
\begin{equation*}
\ket a = \alpha\ket0 +\beta\ket1,\ \ \ \ket b=\gamma \ket0+ \delta\ket1,
\end{equation*}
for some $\alpha,\beta, \gamma, \delta \in \mathbb{C}^2$. Thus
\begin{equation*}
  \ket a \otimes \ket b=\alpha \gamma\ket{00}+\alpha \delta\ket{01}+\beta \gamma\ket{10}+\beta \delta \ket{11}.
\end{equation*}
According to equation (\ref{bla}), this implies
\begin{eqnarray*}
  \alpha \gamma &=& \beta \delta \ = \ \frac{1}{\sqrt2} \\
  \alpha \delta &=& \beta \gamma \ = \ 0,
\end{eqnarray*}
which is an insoluble system. Therefore, we conclude that $\ket{\Phi_+}$ is entangled. (In fact, $\ket{\Phi_+}$ is one of the so-called Bell states, which are maximally entangled states of $\mathbb{C}^2\otimes\mathbb{C}^2$.)
\end{exemplo}

Generalizing the above concepts, we have the following definition.
\begin{definicao} A state $\rho$ of a composite system $\mathcal{H}_A\otimes \mathcal{H}_B$ is said to be separable if \begin{equation*}
  \rho=\sum_i{p_i \rho_{A,i}\otimes \rho_{B,i}}
\end{equation*}
for some $\rho_{A,i} \in \mathcal{H}_A$ and $\rho_{B,i} \in \mathcal{H}_B$, with $\sum_i{p_i}=1$ and $p_i\geq0$. A state which is not separable is said to be entangled.
\end{definicao}

\section{Measurements}\label{mesaurements}

We saw how to represent systems and states. Now we will define how an observer can interact with the systems, that is, how a measurement takes place.

\begin{definicao}Given a state $\rho \in \mathcal{H}$, a measurement on $\rho$ is a set of operators $\{M_i\}$ that acts on $\mathcal{H}$ and satisfy the completeness relation
\begin{equation*}
  \sum_i{M_i^{\dagger}M_i}=I.
\end{equation*}
The index $i$ refers to the outcomes that may occur in the measurement; outcome $i$ occurs with probability
\begin{equation*}
\tr(M_i \rho M_i^{\dagger})
\end{equation*}
and the state of the system after the measurement is
\begin{equation*}
\frac{M_i \rho M_i^{\dagger}}{\tr(M_i \rho M_i^{\dagger})}.
\end{equation*}
\end{definicao}
It is easy to see that the completeness relation implies that the probabilities of the outcomes indeed sum to one.

The outcomes of a measurement can be attached to each measurement operator in an arbitrary way; in general, they do not have a special meaning. For projective measurements however, we will see that this is usually done in a somewhat natural way.

The fact that we use the trace to calculate the probabilities is known as Born's rule. In order to formalize and to strengthen the axiomatic basis of Quantum Theory, Andrew Gleason showed that every probability measure on the set of closed subspaces of a Hilbert space with dimension greater than two is given by the trace of its projector times a positive semi-definite operator with trace one \cite{Gle}. In some sense, this result legitimate the use of the Born's rule.

The next definition can be found in \cite{ABC}.

\begin{definicao}\label{prob meas on C} Let $\mathcal{V}$ be a vector space and $\mathcal{C}$ the set of closed subspaces of $\mathcal{V}$. A probability measure on $\mathcal{C}$ is a measure $\mu:\mathcal{C}\rightarrow[0,1]$ such that $\mu(\mathcal{V})=1$ and if $\{E_i\}\subset\mathcal{C}$ is a countable collection of mutually orthogonal subspaces that generates $E\subset\mathcal{C}$ then
\begin{equation*}
  \mu(E)=\sum_i{\mu(E_i)}
\end{equation*}
\end{definicao}

\begin{teo}[Gleason]
Let $\mathcal{V}$ be a vector space over $\mathbb{C}$ with dimension $d>2$ and $\mathcal{C}$ the set of closed subspaces of $\mathcal{V}$. If $\mu$ is a probability measure on $\mathcal{C}$, then there is a positive semi-definite operator $\rho\in\mathcal{L(V)}$ with unit trace such that
\begin{equation*}
  \mu(E)=\tr(P\rho),
\end{equation*}
for all $E\in\mathcal{C}$, where $P$ is the projector onto $E$.
\end{teo}

The proof of Gleason's theorem can be found in his original work \cite{Gle} and in a simplified version in \cite{Bell66}. In \cite{Ara}, the author discusses other theorems that formalize Born's rule.

We will see later that Gleason's theorem can give us a hint for the construction of local models.

\subsection{Projective measurements}\label{projective}

A simple and important kind of measurements occurs when the measurement operators $\{P_i\}$ are projectors, i.e., $P_i^2=P_i$, and satisfies $P_iP_j=\delta_{ij}P_i$. In this case, the completeness relation becomes
\begin{equation*}
  \sum_i{P_i}=I
\end{equation*}
and the probabilities are given by
\begin{equation*}
  \tr(P_i\rho P_i^{\dagger})=\tr(P_i\rho P_i)=\tr(P_i^2\rho)=\tr(P_i\rho).
\end{equation*}
If the state is pure, then
\begin{equation*}
  \tr(P_i\rho)=\tr(P_i \ket{\psi}\bra{\psi})=\bra{\psi}P_i\ket{\psi},
\end{equation*}
where the last equality is given by Proposition \ref{probpureprojec}.

While a general measurement can be associated with a simple partition of the identity $I$, in the projective case the association can go further, identifying each projective measurement $\{P_i\}$ with the Hermitian operator
\begin{equation*}
  M=\sum_i{m_iP_i}
\end{equation*}
(here in its spectral decomposition form), acting on the space state of the system. In this case, the measurement operators are the projectors onto its eigenvectors and the possible outcomes are its eigenvalues. Such Hermitian operator $M$ is called an \emph{observable}. The set of observables of a system $\mathcal{H}$ will be denoted by $\mathcal{O(H)}\subset\mathcal{L(H)}$.

\begin{exemplo} Important examples of observables are the Pauli matrices,
\begin{equation*}
  \sigma_x=\left(
             \begin{array}{cc}
               0 & 1 \\
               1 & 0 \\
             \end{array}
           \right), \ \sigma_y=\left(
             \begin{array}{cc}
               0 & -i \\
               i & 0 \\
             \end{array}
           \right),\ \sigma_z=\left(
             \begin{array}{cc}
               1 & 0 \\
               0 & -1 \\
             \end{array}
           \right).
\end{equation*}

Together with the identity $I_{2\times2}$, the Pauli matrices span the real vector space $H_2(\mathbb{C})$ of the $2\times2$ Hermitian matrices.

The measurement of observable $\sigma_i$ is referred as ``measurement of spin along the $i$ axis", for $i=x,y,z$.
\end{exemplo}

\begin{exemplo}
More generally, given $\ket v=(v_x, v_y, v_z)$ a real three-dimensional unit vector, we can define the observable\footnote{We denote $\ket a \cdot \ket b$ simply by $a \cdot b$, in order to do not overload the notation.}
\begin{equation*}
  v\cdot\sigma \equiv v_x\sigma_x+v_y\sigma_y+v_z\sigma_z=\left(
             \begin{array}{cc}
               v_z & v_x-iv_y \\
               v_x+iv_y & -v_z \\
             \end{array}
           \right),
\end{equation*}
which has eigenvalues $\pm1$. Hence, a projective measurement related to $v\cdot\sigma$ has possible outcomes $\pm1$ and projectors onto the corresponding eigenspaces are given by
\begin{equation*}
  P_{\pm}=\frac{I\pm v \cdot \sigma}{2}
\end{equation*}
and thus the corresponding probabilities of the measurement applied to the pure state $\ket{\psi}$ are
\begin{eqnarray*}
  \p(\pm1)&=&\tr\left(\ketbraa{\psi}\frac{I\pm v \cdot \sigma}{2}\right) \\
   &=& \frac{1\pm\bra{\psi}v\cdot\sigma\ket{\psi}}{2}.
\end{eqnarray*}

The measurement of this observable is referred to as ``measurement of spin along the $\ket v$ direction".
\end{exemplo}

\begin{exemplo}\label{shaka}
Take the pure state
\begin{equation*}
\ket{+}=\frac{1}{\sqrt2}(\ket0+\ket1)
\end{equation*}
and let's measure the observable $\sigma_z$. As said above, its eigenvalues are $\pm1$ associated to eigenvectors $\ket0, \ket1$, with respective projectors (measurement operators) $P_{\pm}^z=(I\pm\sigma_z)/2$. Hence, the possible resulting states post-measurement are $\ket0$ and $\ket1$, occurring with probability $\tr(P_-^z\ketbraa+)=1/2=\tr(P_+^z\ketbraa+)$.

On the other side, $\ket+$ is itself an eigenvector of the observable $\sigma_x$, together with $\ket-=(\ket0-\ket1)/\sqrt2$. Notice that we can also choose the opposite way and write the ``z-basis" $\{\ket0,\ket1\}$ in terms of the ``x-basis" $\{\ket+,\ket-\}$:
\begin{equation*}
  \ket0=\frac{1}{\sqrt2}(\ket++\ket-), \ \ket1=\frac{1}{\sqrt2}(\ket+-\ket-).
\end{equation*}

Therefore, measuring observable $\sigma_x$ on any of the states $\ket0, \ket1$, the probability to obtain any of the possible resulting states $\ket-, \ket+$ is $\tr(P_-^x\ketbraa+)=1/2=\tr(P_+^x\ketbraa+)$, where the measurement operators are $P_{\pm}^x=(I\pm\sigma_x)/2$.

We conclude from these observations that measuring observable $\sigma_z$ on an eigenvector of $\sigma_x$ leads to a resulting state uniformly distributed on $\{\ket0, \ket1\}$; but measuring $\sigma_x$ on any of those states leads to a resulting state uniformly distributed on $\{\ket+,\ket-\}$. Thus it is impossible for a particle to have spin determined simultaneously in both directions $x$ and $z$. Though it perhaps seems an unimportant conclusion, this fact will play a decisive role in Section \ref{epr}.

\end{exemplo}

\subsection{POVMs}

Observe that, since the probabilities of a measurement $\{M_i\}$ are given by $\tr(M_i\rho M_i^{\dagger})=\tr(\rho M_i^{\dagger}M_i)$, we can associate to the measurement the positive operators $\{Q_i\}$, where $Q_i=M_i^{\dagger}M_i$. Although such operators are not enough to determine the resulting post-measurement state, many times we are only interested in the probabilities $\tr(\rho Q_i)$, and we can restrain our attentions to the simplified description of the measurement provided by the $\{Q_i\}$. These operators are known as the \emph{POVM elements} associated to the measurement (sometimes called \emph{effects}), where the acronym holds for positive operator-value measure.

\section{The Bloch sphere}

The simplest quantum system that we can imagine is the one associated to the Hilbert space $\mathbb{C}^2$, in which the pure states are unit vectors of the form $\ket{v}=\alpha \ket0+\beta \ket1$ and are called \emph{qubits}. As we imply, the vectors $\ket0, \ket1$ form an orthogonal base for $\mathbb{C}^2$.

Because $\ket{v}$ has norm 1, we must have $|\alpha|^2+|\beta|^2=1$. Writing $\alpha=\alpha_1+i\alpha_2$ and $\beta=\beta_1+i\beta_2$, this condition becomes $\alpha_1^2+\alpha_2^2+\beta_1^2+\beta_2^2=1$, so in principle we see that the set of qubits is isomorphic to the sphere $\mathbb{S}^3\subset\mathbb{R}^4$.

However, given an observable of $\mathcal{L}(\mathbb{C}^2)$, its eigenvectors are orthogonal, thus they form a base for $\mathbb{C}^2$. The coefficients of a pure state written in such base are related to the probabilities with which the resulting post-measurement state becomes the corresponding eigenvector.

\begin{exemplo} If $\ket{v}=\alpha \ket0+\beta \ket1$ and we measure $\sigma_z$ (whose eigenvectors are $\ket0$ and $\ket1$), then
\begin{equation*}
  \tr(\ketbraa0 \ketbraa{v})=|\alpha|^2 , \ \tr(\ketbraa1 \ketbraa{v})=|\beta|^2,
\end{equation*}
that is, the resulting state is $\ket0$ with probability $|\alpha|^2$ and $\ket1$ with probability $|\beta|^2$.
\end{exemplo}

Hence, since a state describes the system by giving the probabilities of the possible outcomes of any measurement, we can identify the pure states $\ket{v}$ and $e^{i\theta}\ket{v}$: the $e^{i\theta}$ factor will not influence in the computation of probabilities $|\alpha|^2$ and $|\beta|^2$. Another way to say this is that $\ket{v}$ and $e^{i\theta}\ket{v}$ are physically indistinguishable.

This identification tells us that we have no need of four real coefficients to describe the set of pure states physically distinct, only three will suffice\footnote{The right way to formalize this claim is by showing that the referred identification determines a equivalence relation and considering the quotient space generated by it \cite{ABC}.}. Thus, the set of physically distinct qubits is isomorphic to $\mathbb{S}^2$.

As we saw in Section {\ref{projective}}, to each real three-dimensional unit vector $\ket v$ can be associated the observable $v\cdot \sigma$. Now we know that each of such vectors can also be associated uniquely to a pure state. In order to distribute the pure states $\ket v$ in $\mathbb{S}^2$, we can take each of them to correspond to the direction of the eigenvector associated to +1 in the measurement of the observable $v \cdot \sigma$. The sphere $\mathbb{S}^2$ together with this correspondence is called \emph{the Bloch sphere}.

Notice that the pair of vectors that are orthogonal in $\mathbb{R}^3$ are collinear in the Bloch sphere: the states corresponding to vectors $\ket0, \ket1$ points to the north and south poles, and $\ket-, \ket+$ are in opposite points of the equator of the Bloch sphere. We can think that, in the process of allocation of states corresponding to vectors from $\mathbb{R}^3$ in the Bloch sphere, the angles between the vector and the positive $y$-axis get doubled. This mean that if we consider a scalar product over the Bloch sphere, it will be given by
\begin{equation*}
\ket a \cdot \ket b= \cos(2\theta),
\end{equation*}
where $\theta$ is the angle between the vectors in $\mathbb{S}^2$.

\begin{figure}[h!]
  \centering
  \includegraphics[width=15cm]{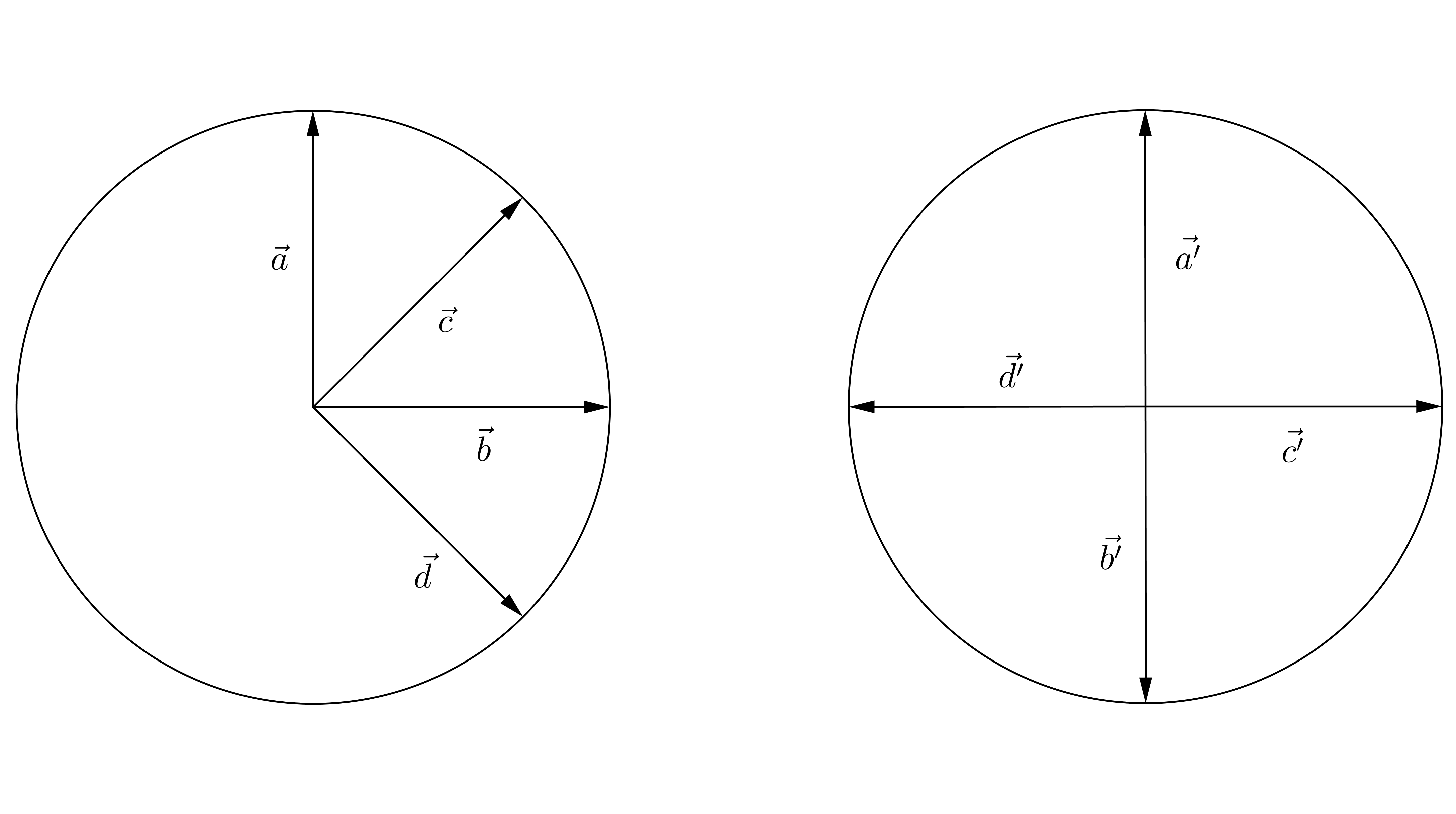}\\
  \caption{\small{Suppose the vectors $\vec{a},\vec{b}$ are $\ket0, \ket1$ and vectors $\vec{c},\vec{d}$ are $\ket+,\ket-$. In the left side we have them represented in a section of the sphere $\mathbb{S}^2$ and in the right, in a section of the Bloch sphere.}}
\end{figure}

\section{The partial trace}

\begin{definicao} Let $\mathcal{H_A}\otimes \mathcal{H_B}$ be a composite Hilbert space. We define the partial trace (in relation to $\mathcal{H_B}$) $\tr_B:\mathcal{L(H_A)}\otimes \mathcal{L(H_B)}\rightarrow \mathcal{L(H_A)}$ by
\begin{equation*}
  \tr_B(M_A\otimes M_B)=M_A\cdot \tr(M_B)
\end{equation*}
for product operators and extend to non-product operators by linearity. Analogously, we define the partial trace in relation to subsystem $\mathcal{H_A}$.
\end{definicao}

Given a state of a composite system $\mathcal{H_A\otimes H_B}$, we can find the state that describe each subsystem via partial trace. Such state is called \emph{reduced density operator}.

\begin{definicao} Let $\rho$ be the state that describes the composite system $\mathcal{H_A\otimes H_B}$. Then the reduced density operator
\begin{equation*}
  \rho_A=\tr_B(\rho)
\end{equation*}
describes subsystem $\mathcal{H_A}$.
\end{definicao}

The description provided by the partial trace referred above is about measurements: suppose that Alice shares the state $\rho$ with Bob and perform a measurement $\{M_i\}$. Then $\tr_B(\rho)=\rho_A$ is the only state that satisfies
\begin{equation*}
  \tr(M_a\otimes I \rho)=\tr(M_a \rho_A)
\end{equation*}
for any $a$ and $\{M_i\}$, that is, is the only state that provides the right probability of obtain outcome $a$ when we consider a measurement only on Alice's system. In other words, concerning to local measurements only on Alice's side, to say that Alice shares state $\rho$ with Bob is the same as to say that Alice holds the state $\tr_B(\rho)$.

\begin{exemplo}\label{aldebaran} Suppose that Alice and Bob shares the singlet state $\ketbraa{\Psi_-}$. Then Alice's reduced state is
\begin{eqnarray*}
  \tr_B(\ketbraa{\Psi_-}) &=& \frac{1}{2}\tr_B(\ketbraa{01}+\ketbraa{10}-\kett01\braa10- \kett10\braa01) \\
   &=& \frac{1}{2}(\ketbraa0\tr(\ketbraa1)+\ketbraa1\tr(\ketbraa0)- \ketbra{0}{1}\tr(\ketbra{1}{0})\\ &&-\ketbra{1}{0}\tr(\ketbra{0}{1})) \\
   &=& \frac{1}{2}(\ketbraa{0}+\ketbraa1)\\
   &=& \frac{I_{2\times2}}{2}.
\end{eqnarray*}
\end{exemplo}

\section{EPR}\label{epr}

In objection to the view where measurements in quantum systems do not properly reveal some already existing information, but rather just lead the system to probabilistically assume one of the possible resulting states, Albert Einstein, together with Nathan Rosen and Boris Podolsky proposed a thought experiment \cite{EPR}. Their idea was to show, using entangled states, that the description of reality provided by Quantum Theory was not complete. Here, we are going to use the simplified formulation given by David Bohm \cite{Bohm}.

The first step is to understand what the authors propose as description of reality.

\begin{definicao}\label{shura} An element of reality is a physical property that can be determined precisely (i.e., with probability equal to unity) without disturbing the system. A theory is complete if contains a counterpart for every element of reality.
\end{definicao}

EPR consider a bipartite system where the parts, Alice and Bob, share the singlet state
\begin{equation*}
  \ket{\Psi_-}=\frac{1}{\sqrt2}(\kett01 -\kett10).
\end{equation*}
The singlet state has the property\footnote{This property is derived from the fact that the singlet has total spin equal to zero, which roughly means that each part must have its spin pointing to exactly opposite directions \cite{Bohm}.} that, given any direction $\ket v$, it can be written as
\begin{equation*}
  \ket{\Psi_-}=\frac{1}{\sqrt2}(\kett{v_+}{v_-} -\kett{v_-}{v_+}),
\end{equation*}
where $\ket{v_+}, \ket{v_-}$ are eigenstates of the observable $v\cdot\sigma$, related to eigenvalues $\pm1$. This means that the outcomes of any measurement of $v\cdot\sigma$ in both parts will be perfectly anti-correlated: if Alice obtains +1, then Bob obtains -1, and vice-versa.

Now imagine that Alice and Bob are far away from each other, Alice measures $\sigma_z$ and obtains outcome +1. So, if Bob also measures $\sigma_z$ in his particle, we know that the resulting outcome will be -1. Alice's measurement cannot instantaneously disturb Bob's system because of the distance separating them, but still we can predict with probability 1 the value of $\sigma_z$. Thus, according to Definition \ref{shura}, $\sigma_z$ is an element of reality of Bob's system.

We can imagine that, instead of $\sigma_z$, Alice measures $\sigma_x$ and using the same argument conclude that $\sigma_x$ is an element of reality of Bob's system as well. (Does not really matter if Alice indeed measured $\sigma_z$ or $\sigma_x$ or did not do any measurement at all: the important fact is that it would be possible to predict the property `spin along the $z$-axis' or `spin along the $x$-axis' of Bob's particle.) Therefore, there exist elements of reality corresponding to those properties simultaneously. A complete theory should, therefore, simultaneously assign values for such properties.

However, we saw in Example \ref{shaka} that Quantum Theory forbids both observables to be determined at the same time\footnote{This rest upon the fact that $\sigma_z$ and $\sigma_x$ do not commute. Any other pair of non-commuting observables would do the job, as we will consider in Chapter \ref{degorre}.}.The conclusion of EPR was that Quantum Theory is not complete. Despite being a good approximation of experimental data, they believed that should exist another theory, satisfying their completeness criterium without paradoxes.

It is important to emphasize that we have taken two major assumptions to insure this conclusion: the \emph{realism} (the intrinsic existence of values for elements of reality, independent of measurements) and the \emph{locality} (distant measurements cannot influence each other instantaneously). Hence, the EPR argument proves that \emph{if} Quantum Theory were a locally realistic theory, \emph{then} would not be a complete one.

In the next chapter, we will see that Quantum Theory is not locally realistic (or, as we will simply say from now on, \emph{local}).

\chapter{Nonlocality}\label{nonlocality}

Nearly thirty years after Ref. \cite{EPR} was published, John Bell made a proposal that allows one to experimentally test whether Quantum Theory is locally realistic or not \cite{Bell64}. The main tool used for performing such test was a linear inequality for the joint probabilities (or equivalently, the joint expectations) of both parts, the first of a class of inequalities now known as \emph{Bell inequalities}. To obey the Bell inequalities is a condition that a local-realist theory should satisfy. As we will see, Quantum Theory is able to violate such inequalities, consequently annulling the incompleteness' conclusion of EPR.

Though Bell have developed the first Bell inequality in 1964, the most popular one is the CHSH inequality, dating from 1969.

\section{The CHSH inequality}\label{belldesi}

The CHSH inequality \cite{CHSH}, named after its authors, Clauser, Horne, Shimony and Holt, is the only one that we will have to keep in mind for the discussions proposed in this thesis.

Suppose that we have two parts, Alice and Bob, and a referee, a third part which is not related to any system, but is able to prepare two particles and distribute them to Alice and Bob, repeating the procedure an unlimited number of times. Once each part has its particle, it must choose among two possible measurements, say $M_Q$ or $M_R$ for Alice and $M_S$ or $M_T$ for Bob. They do not know in advance which measurement they will perform; by receiving the particle, each of the two decides it randomly. Also, each measurement $M_i$ has outcome $i$, for $i=Q,R,S,T$, that takes values on $\{1,-1\}$.

Making the assumption of realism, we will consider that $Q$ (and similarly for $R,S,T$) is an objective property of Alice's particle, being merely revealed by the measurement $M_Q$. The course of action is arranged in such a way that Alice and Bob do their measurements at the precise same time. Thus, assuming locality, Alice's measurement cannot disturb the result of Bob's measurement, and vice-versa.

We will now consider the quantity
\begin{equation*}
QS+RS+RT-QT.
\end{equation*}
Since the above expression can be rewritten as
\begin{equation*}
  Q(S-T)+R(S+T)
\end{equation*}
and $S,T\in\{1,-1\}$, we must have $Q(S-T)=0$ or $R(S+T)=0$. In either case, we have
\begin{equation*}
QS+RS+RT-QT\in\{+2,-2\}.
\end{equation*}

Now, if $\p(q,r,s,t)$ is the probability that right before the measurements are performed the system is in a state that $Q=q, R=r, S=s$ and $T=t$, then the expectation becomes
\begin{eqnarray*}
  \E(QS+RS+RT-QT) &=& \sum_{q,r,s,t\in\{+1,-1\}}{\p(q,r,s,t)(qs+rs+rt-qt)} \\
   &\leq& 2\sum_{q,r,s,t\in\{+1,-1\}}{\p(q,r,s,t)}\\
   &=& 2.
\end{eqnarray*}

On the other hand, expectation is linear,
\begin{equation*}
\E(QS+RS+RT-QT)=\E(QS)+\E(RS)+\E(RT)-\E(QT).
\end{equation*}

Combining the previous equations, we obtain the CHSH inequality \cite{CHSH}
\begin{equation}\label{CHSH}
\E(QS)+\E(RS)+\E(RT)-\E(QT)\leq2.
\end{equation}

With this inequality in hands, Alice and Bob are able to compare whether or not the expected values predicted by their theory obeys the inequality. Depending on the result of the comparison (more specifically, if the inequality is violated), they conclude that their theory is nonlocal. They can also proceed to perform a series of measurements in the above conditions on a shared state. The outcomes obtained by each part can then be put together, and the integrated data can be used to approximate each of the expected values in question. If they get an inequality violation, they will be experimentally confirming the nonlocality of Nature itself. If they check that the data obtained match their predictions, they will be showing that Nature agrees with their nonlocal theory.

We now proceed to prove that Quantum Theory is not local.

\begin{exemplo}\label{masc da mort}
Suppose that Alice and Bob share the singlet state
\begin{equation*}
\ket{\Psi_-}=\frac{\kett01-\kett10}{\sqrt2}.
\end{equation*}
Consider now the observables
\begin{equation*}
  M_Q=\sigma_z,\ M_R=\sigma_x
\end{equation*}
for Alice and
\begin{equation*}
M_S=\frac{-\sigma_z-\sigma_x}{\sqrt2},\ M_T=\frac{\sigma_z-\sigma_x}{\sqrt2}
\end{equation*}
for Bob.

The observables $M_Q, M_S$ have eigenvalues $+1, -1$, with respective eigenvectors $\ket0, \ket1$ for $M_Q$ and $\ket0+(-1-\sqrt2)\ket1, \ket0+(-1+\sqrt2)\ket1$ for $M_S$. Their projectors onto the eigenspaces are
\begin{eqnarray*}
  P_+^Q&=&\ketbraa0 \\ P_-^Q&=&\ketbraa1 \\
  P_+^S&=&\frac{\ketbraa0+(3+2\sqrt2)\ketbraa 1 +(-1-\sqrt2)[\ketbra{0}{1}+\ketbra{1}{0}]\ketbraa1}{4+2\sqrt2} \\ P_-^S&=&\frac{\ketbraa0+(3-2\sqrt2)\ketbraa 1 +(-1+\sqrt2)[\ketbra{0}{1}+\ketbra{1}{0}]\ketbraa1}{4-2\sqrt2}.
\end{eqnarray*}

Hence, Quantum Theory predicts that
\begin{eqnarray*}
  \E(QS) &=& \p(qs=1)-\p(qs=-1) \\
   &=& [\p(q=1,s=1)+\p(q=-1, s=-1)]-[\p(q=+1,s=-1)\\&& +\p(q=-1,s=+1)] \\
   &=& [\tr(P_+^Q\otimes P_+^S\ketbraa{\Psi_-})+\tr(P_-^Q\otimes P_-^S\ketbraa{\Psi_-})]\\ &&- [\tr(P_+^Q\otimes P_-^S\ketbraa{\Psi_-})+\tr(P_-^Q\otimes P_+^S\ketbraa{\Psi_-})]\\
   &=& \left[\frac{3+2\sqrt2}{4(2+\sqrt2)}+\frac{1}{4(2-\sqrt2)}\right]- \left[\frac{3-2\sqrt2}{4(2-\sqrt2)}+\frac{1}{4(2+\sqrt2)}\right]\\
   &=& \frac{1}{\sqrt2}.
\end{eqnarray*}

Similarly, we may find that
\begin{equation*}
  \E(RS)=\E(RT)=\frac{1}{\sqrt2}, \ \E(QT)=-\frac{1}{\sqrt2}.
\end{equation*}
Thus,
\begin{equation*}
  \E(QS)+\E(RS)+\E(RT)-\E(QT)=2\sqrt2>2.
\end{equation*}

Since the CHSH inequality is violated, we conclude that Quantum Theory is nonlocal.

\end{exemplo}

\section{The Horodecki Criterion}\label{horocrit}

Notice that the CHSH inequality (and similarly, all other Bell inequalities) can be seen as a superoperator, whose entries are a quantum state and four observables. In order to test a violation, even after fixing a state, we may achieve different values for the expectation by varying the set of observables.

In Example \ref{masc da mort}, we saw that, for the singlet state, there exist observables $M_Q, M_R, M_S, M_T$ for which the CHSH is violated. This, however, is not the case for all quantum states. In 1995, Ryszard Horodecki, Pawel Horodecki and Michal Horodecki presented a necessary and sufficient condition for an arbitrary quantum state of $\mathbb{C}^2\otimes\mathbb{C}^2$ to violate the CHSH inequality and made explicit the greatest value achieved by that state \cite{HHH}.

Associated to the CHSH inequality, there is the observable
\begin{equation*}
\mathcal{B}_{CHSH}(x,x',y,y')=x\cdot \sigma \otimes (y+y')\cdot\sigma + x'\cdot\sigma\otimes(y-y')\cdot\sigma,
\end{equation*}
called the CHSH \emph{Bell operator}, where $\ket x, \ket{x'}, \ket y, \ket{y'}$ are unit vectors in $\mathbb{R}^3$. The inequality in question thus becomes
\begin{equation*}
\E_{\rho}(\mathcal{B}_{CHSH})\leq2,
\end{equation*}
where the subscript emphasize the referred state. Our objective is to check if $\rho$ violates it, hence we want to maximize the expectation over all Bell operators and show that such quantity is achieved by some Bell operator $\mathcal{B}_{max}$:
\begin{equation*}
\sup_{\mathcal{B}_{CHSH}}{\E_{\rho}(\mathcal{B}_{CHSH})}=\E_{\rho}(\mathcal{B}_{max}).
\end{equation*}

The Horodecki start by showing that for all $\rho\in\mathcal{D}(\mathbb{C}^2\otimes\mathbb{C}^2)$, there are vectors $\ket r, \ket s \in \mathbb{R}^3$ such that
\begin{equation*}
\rho=\frac{1}{4}\left( I\otimes I + (r\cdot\sigma)\otimes I + I\otimes (s\cdot \sigma) + \sum_{m,n=1}^3{t_{nm}\sigma_n \otimes \sigma_m} \right),
\end{equation*}
where $\sigma_x=\sigma_1, \sigma_y=\sigma_2$ and $\sigma_z=\sigma_3$ are the Pauli matrices and $t_{nm}=\tr(\rho\sigma_n\otimes\sigma_m)$.

Consider the $3\times3$ matrix $T_{\rho}$ composed by the coefficients $t_{mn}$. A simple calculation gives us
\begin{equation*}
\E_{\rho}(\mathcal{B}_{CHSH})=\bra x T_{\rho}(\ket y+ \ket y')+\bra {x'} T_{\rho}(\ket y- \ket {y'}).
\end{equation*}

By noticing that the vectors $\ket y+ \ket {y'}, \ket y- \ket {y'}$ are orthogonal and satisfy $||\ket y+ \ket {y'}||^2+||\ket y- \ket {y'}||^2=4$, where $||\cdot||$ is the Euclidian norm, we can introduce the pair of unit and mutually orthogonal vectors $\ket z, \ket {z'}$ such that
\begin{equation*}
\ket y+ \ket {y'}=2\cos\theta \ket z,\ \ \ \ket y- \ket {y'}=2\sin\theta\ket {z'},
\end{equation*}
for some $\theta \in [0,\pi]$. Then we have
\begin{eqnarray}\label{saori}
  \sup_{\mathcal{B}_{CHSH}}{\E_{\rho}(\mathcal{B}_{CHSH})} &=& \sup_{\theta, \ket x, \ket{x'}, \ket z, \ket{z'}}{2\left[\bra x T_{\rho}\ket z \cos\theta + \bra {x'} T_{\rho}\ket {z'} \sin\theta\right]} \\
   &=& \sup_{\theta, \ket z, \ket{z'}}{2\left[||T_{\rho}\ket z|| \cos\theta + ||T_{\rho}\ket {z'}|| \sin\theta\right]} \\
   &=& \sup_{\ket z, \ket{z'}}{2\sqrt{||T_{\rho}\ket z||^2 + ||T_{\rho}\ket {z'}||^2}}
\end{eqnarray}

Now, consider the positive matrix $U_{\rho}=T_{\rho}^{\dagger}T_{\rho}$ and denote its two greatest eigenvalues by $u, \tilde u$. Using Lagrange multipliers, it is possible to show that
\begin{equation*}
\sup_{\{\ket j,\ket k\}}{(||T_{\rho}\ket j||^2 + ||T_{\rho}\ket k||^2)}=u + \tilde u =: M(\rho),
\end{equation*}
where the supreme is taken over all orthonormal subsets $\{\ket j,\ket k\}\subset \mathbb{R}^3$. Eq. (\ref{saori}) shows us that
\begin{equation*}
\sup_{\mathcal{B}_{CHSH}}{\E_{\rho}(\mathcal{B}_{CHSH})}=2\sqrt{M(\rho)}.
\end{equation*}

Conversely, one can take in turn $\ket {z_{max}}, \ket{z'_{max}}$ as the $U_{\rho}$ eigenvectors maximalizing $M(\rho)$; $\ket {x_{max}}, \ket{x'_{max}}$ as unit vectors in the directions $T_{\rho}\ket{z_{max}}, T_{\rho}\ket{z'_{max}}$; and $\theta_{max}$ defined by $||T_{\rho}\ket{z_{max}}||\sin\theta_{max}=||T_{\rho}\ket{z'_{max}}||\cos\theta_{max}$. Using these parameters to construct the observable $\mathcal{B}_{max}$, we have
\begin{equation*}
2\sqrt{M(\rho)}=\E_{\rho}(\mathcal{B}_{max})=\sup_{\mathcal{B}_{CHSH}}{\E_{\rho}(\mathcal{B}_{CHSH})}.
\end{equation*}

The above reasoning outlines the proof of the following result.

\begin{teo} There is a CHSH Bell operator $\mathcal{B}_{max}$ such that
\begin{equation*}
\E_{\rho}(\mathcal{B}_{max}) = \sup_{\mathcal{B}_{CHSH}}{\E_{\rho}(\mathcal{B}_{CHSH})} = 2\sqrt{M(\rho)}.
\end{equation*}
\end{teo}

The straightforward corollary yielded is the announced Horodecki Criterion.

\begin{coro}[Horodecki Criterion] A state $\rho\in\mathcal{D}(\mathbb{C}^2\otimes\mathbb{C}^2)$ violates the CHSH inequality if and only if $M(\rho)>1$.
\end{coro}

\section{Local hidden variables}

The EPR argument can be generalized to a ``Bell scenario", by considering an arbitrary number of systems, possible measurements and possible outcomes. What is central in such kind of experiment is that - let's suppose the number of systems is two - Alice and Bob both receive from a referee a shared state and give rise to joint probabilities $\p(ab|x,y)$, where $a(b)$ is a possible outcome for $x(y)$, one of the possible measurements available to be performed by Alice (Bob). As we saw, in general we have
\begin{equation*}
\p(a,b|x,y)\neq\p(a|x)\p(b|y),
\end{equation*}
that is, the outcomes on both sides are not always independent from each other. The existence of such correlations does not necessarily imply a direct influence of one system on the other, though. It can be the case that it is only being revealed a dependence relation between the two systems established in the past, when they may have interacted \cite{BCPSW}.

If that is the case, there is a set of past factors, described by some variables $\lambda$, which influences the outcomes, explaining completely the dependence between them. This means that we must be able to write
\begin{equation}\label{afrodite}
\p(a,b|x,y,\lambda)=\p(a|x,\lambda)\p(b|y,\lambda),
\end{equation}
standing for the fact that the only needed information to determine the probability of outcome $a$ occurring is the local measurement $x$ in question and the past variables $\lambda$; the distant measurement performed by Bob has nothing to do with it.

Since we are not claiming to have any knowledge over $\lambda$ besides its existence, in principle it may involve physical quantities that are not completely controllable. Thus, $\lambda$ will not necessarily be constant for each run of the experiment, even if the referee is careful enough to use the exact same process to prepare the states. Thus, we must consider also a probability distribution $\pi(\lambda)$ over $\Lambda$, the set where these variables inhabit, that rules the different values of $\lambda$ for different experiments.

Eq. (\ref{afrodite}) together with this considerations lead to
\begin{equation}\label{saga}
\p(a,b|x,y)=\int_{\Lambda}{\pi(\lambda)\p(a|x,\lambda)\p(b|y,\lambda)}.
\end{equation}

Another observation that should be made is that the distribution $\pi(\lambda)$ is supposed to be independent of the chosen measurements, i.e., for any choice of $x,y$, the $\lambda$ variables are distributed in the same way. In other words, $\pi(\lambda)$ must \emph{a priori} take in consideration all possible measurements on the system.

By using Eq. (\ref{saga}), it is possible to rewrite the same expectations of Example \ref{masc da mort} in terms of $\lambda$, and then derive the CHSH inequality (\ref{CHSH}), as is done in \cite{BCPSW}. That is, the assumption of this past variables $\lambda$ can play the part of the local realism assumption, formulated by EPR. Based on that, we can take Eq. (\ref{saga}) as the condition for locality. The variables $\lambda$ are called \emph{local hidden variables}: `local' because their knowledge enable us to make a local factorization of the joint probabilities, `hidden' because they are apart from the state\footnote{We are considering the terms locality and local realism to have basically the same meaning. For a discussion on those terminologies, check \cite{Nor}.}.

\section{Local models}\label{local models}

Strictly speaking, locality is a feature of families of probabilities distributions $\{\p(a,b|x,y)\in\mathbb{R}| a\in A, b\in B, x\in X, y\in Y\}$, where $A,B$ are the sets of possible outcomes and $X,Y$ the sets of possible measurements of each part. So, when we say that Quantum Theory is nonlocal, what we really mean is that there exists at least one bipartite state and one pair of measurements that provide a set of correlations which cannot be written in the factorized manner of Eq. (\ref{saga}), even when local hidden variables are taken in account. This is guaranteed by a violation of a Bell inequality, like we saw in Example \ref{masc da mort}.

However, this is not the case for all quantum states. Indeed, for product states the correlations factorize naturally.

\begin{prop}\label{kamuss} Let $A,B \in \mathcal{L(H)}$. Then $\tr(A\otimes B)=\tr(A)\tr(B)$.
\end{prop}
\begin{proof}
Since\footnote{Check the proof of Proposition \ref{kamus}.}
\begin{equation*}
  A\otimes B=\sum_{i,j,k,l}{a_{ij}b_{kl}\ketbra{ki}{lj}},
\end{equation*}
we have
\begin{eqnarray*}
  \tr(A\otimes B) &=& \sum_{i,j,k,l}{a_{ij}b_{kl}\tr(\ketbra{ki}{lj})}\\
  &=& \sum_{i,k}{a_{ii}b_{kk}} \\
  &=& \sum_i{a_{ii}}\sum_k{b_{kk}} \\
  &=& \tr(A)\tr(B).
\end{eqnarray*}
\end{proof}

With Proposition \ref{kamuss} in mind, we can easily see that a measurement on a product state $\rho^A \otimes \rho^B$ can always be locally factorized:
\begin{equation*}
\p(a,b|\{M_i\},\{N_j\})=\tr(M_a\rho^A \otimes N_b\rho^B)=\tr(M_a\rho^A)\tr(N_b\rho^B)=\p(a|\{M_i\})\p(b|\{N_j\}).
\end{equation*}
It is as if no local variables were hidden, or if those who were make no difference in the outcomes: $\p(a|x,\lambda)=\p(a|x)$ (and the same for Bob). Thus the locality condition is trivially achieved: for any measurements $x,y$ we have
\begin{eqnarray*}
  \p(a,b|x,y) &=& \p(a|x)\p(b|y) \\
   &=& \int_{\Lambda}{\pi(\lambda)\p(a|x)\p(b|y)d\lambda} \\
   &=&\int_{\Lambda}{\pi(\lambda)\p(a|x,\lambda)\p(b|y,\lambda)d\lambda}.
\end{eqnarray*}

On the other side, there is no reason to exclude the existence of an entangled state for which the locality condition is satisfied as well. In principle, it may be that the relation between the systems represented by the entanglement can also be explained by some set of local hidden variables. This would mean that, setting good choices of objects to play the part of the hidden variables $\lambda\in\Lambda$, a probability distribution $\pi(\lambda)$ of these objects and a manner to compute the probabilities of the outcomes generated by each part, we would be able to artificially simulate the correlations $\p(ab|x,y)$ in a way that Eq. (\ref{saga}) would always be satisfied, no matter which measurements $x,y$ were being considered. In the case where we manage to perform all these tasks successfully, then we say that we have created a \emph{local hidden variables model} for such state.

\begin{definicao} Let $\mathcal{M(H)}=\{\{M_i\}\subset\mathcal{L(H)};\ \sum_i{M_i^{\dagger}M_i}=I\}$ be the set of generalized measurements over $\mathcal{H}$, $O_M\subset\mathbb{R}$ the set of possible outcomes for the measurement $M=\{M_i\}$ and $\Lambda$ the set of hidden variables. A function
\begin{eqnarray*}
f^M:O_M\times\Lambda\times\{\{M_i\}\} &\rightarrow&[0,1]\\
(a,\lambda, \{M_i\})&\mapsto& f^M(a,\lambda, \{M_i\})
\end{eqnarray*}
is said to be a response function for $M$ if for every fixed $\lambda\in\Lambda$ we have
\begin{equation*}
\sum_{j\in O_M}{f^M(j,\lambda,\{M_i\})}=1.
\end{equation*}
A function $f:\mathbb{R}\times\Lambda\times\mathcal{M(H)} \rightarrow[0,1]$ is said to be a response function if every restriction
\begin{eqnarray*}
f|_M:O_M\times\Lambda\times\{\{M_i\}\} &\rightarrow& [0,1]\\
(a,\lambda,\{M_i\}) &\mapsto& f(a,\lambda,\{M_i\})
\end{eqnarray*}
is a response function for $M$, that is,
\begin{equation*}
f|_M=f^M,\ \forall M\in\mathcal{M(H)}.
\end{equation*}
\end{definicao}

The response functions depend on the outcomes, the hidden variables and the whole measurement being applied\footnote{Indeed, if the response functions depend only on the specific measurement operator regarding the outcome of interest, Gleason's theorem would imply that the correlations being reproduced would belong to a separable state. See Section \ref{alice response function} for further details.}, since their task is to attribute a probability for each outcome.

\begin{definicao}\label{def local model}A state $\rho\in \mathcal{H}_A\otimes \mathcal{H}_B$ admits a local hidden variables model for projective measurements if there exists a set $\Lambda$ of hidden variables, a probability distribution $\pi(\lambda)$ and response functions $\p_A, \p_B$ such that for any pair of observables $P=\sum_i{iP_i},\ Q=\sum_j{jQ_j}$ measured by Alice and Bob, respectively, the equality
\begin{equation}\label{local}
  \tr(\rho P_a\otimes Q_b) =\int_{\Lambda}{\pi(\lambda)\p_A(a,\lambda,\{P_i\})\p_B(b,\lambda,\{Q_j\})d\lambda},
\end{equation}
is attained.
\end{definicao}

Notice that left hand side of the above equality stands for $\p(a,b|P_a,Q_b)$, so we have there a twin of Eq. (\ref{saga}), where $\p_A(a,\lambda,P)=\p(a|\lambda, P)$ and $\p_B(b,\lambda,Q)=\p(b|\lambda, Q)$.

There is no reason for us to restrain the idea of locality to projective measurements. The definition below is a stronger version of Def. \ref{def local model}, since projectors are a particular case of positive operators.

\begin{definicao} Analogously to Def. \ref{def local model}, a state $\rho\in \mathcal{H}_A\otimes \mathcal{H}_B$ admits a local hidden variables model for POVMs if Eq. (\ref{local}) is satisfied for any pair of POVMs $\{M_i\}, \{N_j\}$, that is,
\begin{equation*}
  \tr(\rho M_a\otimes N_b) =\int_{\Lambda}{\pi(\lambda)\p_A(a,\lambda,\{M_i\})\p_B(b,\lambda,\{N_j\})d\lambda}.
\end{equation*}

\end{definicao}

We will usually use the shortcut expression `local model' to refer to local hidden variables model. In the same fashion, sometimes we will refer to a state which admits a local model as simply `local'.

The locality of product states discussed above naturally leads to a local model for a separable (mixed) state. If our state is
\begin{equation*}
\rho=\sum_{i=1}^{n}{p_i\rho^A_i\otimes \rho^B_i}
\end{equation*}
then the probabilities we would like to reproduce have the form
\begin{equation*}
\tr(M_a\otimes N_b\rho)=\sum_{i=1}^{n}{p_i\tr(M_a\rho^A_i)\tr(N_b\rho^B_i)}.
\end{equation*}
The comparison of the integral in Eq. (\ref{local}) and the sum in the right side of the above equation practically solve the problem: we just have to take the hidden variables as $\lambda\in\Lambda=\{1,...,n\}$ and $\pi(\lambda)=p_{\lambda}$. The response functions will be same for both parts, the ``quantum" response function given by the trace of the state times the measurement operator. Hence, the role of the hidden variables here is only to determine which of the product states present in the mixture we are going to use in the response function.

Since separability implies locality, by contraposition, we conclude
\begin{equation*}
  nonlocality \Rightarrow entanglement.
\end{equation*}

In \cite{Gis}, Nicolas Gisin showed that for every pure bipartite entangled state, it is possible to construct observables for which the CHSH inequality is violated. Later, in Ref. \cite{PR} and \cite{YCZLO} it has been proved that this extends to every multipartite pure state. In another words, for pure states we have the equivalence
\begin{equation*}
  nonlocality \iff entanglement.
\end{equation*}

The great and perhaps counterintuitive advance acquired by Reinhard Werner was to show that for mixed states the situation is not the same \cite{Wer}. That is, it does exist entangled states which admit local models, and thus violate no Bell inequality! So, in general, we have
\begin{equation*}
entanglement \nRightarrow nonlocality.
\end{equation*}

That is the first big result that we are going to present, in the next chapter.

\chapter{Werner's local model}\label{Werner}

This whole chapter is devoted to detail Werner's paper \cite{Wer}, in which the first local model appears.

There are two main difficulties in constructing an example of entangled state which admits a local model for projective measurements. The first is to prove that such state is indeed entangled, i.e., it can not be written as a convex combination os separable states. The second is to verify that Eq. ($\ref{local}$) holds independently of $P,Q$, which become an infinite system of equations indexed by the set of projective measurements. We shall circumvent both difficulties by considering states of very high symmetry, nominated Werner states.

\section{Werner states}

\begin{definicao} A state $W\in \mathcal{D}(\mathcal{H} \otimes \mathcal{H})$ is said to be a \emph{Werner state} if it is $U\otimes U$-invariant, i.e., $U\otimes U W U^{\dagger} \otimes U^{\dagger}=W,\ \forall U\in \mathcal{L}(\mathcal{H})$ unitary.
\end{definicao}

The $U\otimes U$-invariance can be physically interpreted as the property that allows both Alice and Bob to apply the same unitary transformation in their part of the system and still remain with the same global state.

In order to characterize the Werner states, we start asking which operators, not necessarily states, present the $U\otimes U$-invariance. It is clear, by the definition of unitary, that the identity $I$ has this property. For the \emph{flip} operator $V:\mathcal{H} \otimes \mathcal{H} \rightarrow \mathcal{H} \otimes \mathcal{H}$, defined by
\begin{equation*}
V\ket{ij}=\ket{ji}
\end{equation*}
on the product states and extended by linearity, we have $V=V^{-1}=V^{\dagger}$. It follows that
\begin{equation*}
V A\otimes BV^{\dagger}=B\otimes A
\end{equation*}
and thus
\begin{eqnarray*}
  (U\otimes U) V (U^{\dagger} \otimes U^{\dagger}) = V \ \Leftrightarrow \ (U\otimes U) V (U^{\dagger} \otimes U^{\dagger}) V^{-1} = I \\ \Leftrightarrow  \ (U\otimes U) (V U^{\dagger} \otimes U^{\dagger} V^{\dagger}) = I \ \Leftrightarrow \  (U\otimes U) (U^{\dagger} \otimes U^{\dagger}) = I,
\end{eqnarray*}
so $V$ is also $U\otimes U$-invariant. By linearity, we have that any linear combination of these two operators presents this property. The first important result about Werner states is that those are all the operators that have the $U\otimes U$-invariance.

\begin{prop} An operator $A\in \mathcal{L}(\mathbb{C}^d \otimes \mathbb{C}^d)$ is $U\otimes U$-invariant if and only if $A$ is a linear combination of the identity $I$ and the flip operator $V$.
\end{prop}

\begin{proof} One of the implications was established above. To the other way, consider $A$ $U\otimes U$-invariant and the orthonormal basis $\mathcal{B}=\{\ket{e_1 e_1}, \ket{e_1 e_2}, ..., \ket{e_d e_d}\}$ of $\mathbb{C}^d \otimes \mathbb{C}^d$. The matrix of $A$ associated to this basis has entries  $\bra{e_j e_k}A\ket{e_l e_m}$. Using the invariance of $A$ under unitaries $U_r$ which take $\ket{e_r}$ to $-\ket{e_r}$ leaving the other elements fixed, we have
\begin{equation*}
  \bra{e_j e_k}A\ket{e_l e_m}=\bra{e_j e_k}U_r^{\dagger}\otimes U_r^{\dagger}A U_r\otimes U_r\ket{e_l e_m}= \bra{U_r\otimes U_r e_j e_k} A U_r\otimes U_r\ket{e_l e_m}
\end{equation*}
for $r=1,...,d$. Such matrix elements vanishes unless the indices $j,k,l,m$ are equal by pairs. Indeed, if we have $j\notin \{k,l,m\}$, then
\begin{equation*}
  \bra{e_j e_k}A\ket{e_l e_m}=\bra{U_j\otimes U_j e_j e_k} A U_j\otimes U_j\ket{e_l e_m}=-\bra{e_j e_k}A\ket{e_l e_m}.
\end{equation*}

In the same way, making use of the unitaries $\tilde{U}_r$ that take $\ket{e_r}$ to $i\ket{e_r}$ and leaves the other elements fixed, we see that the terms of the form $\bra{e_j e_j}A\ket{e_l e_l}$ also vanishes.

We conclude that the only non-zero matrix elements have the form $\bra{e_j e_j}A\ket{e_j e_j}$, $\bra{e_j e_k}A\ket{e_j e_k}$ or $\bra{e_j e_k}A\ket{e_k e_j}$. Since any transposition of two of the basis elements can be realised unitarily, we can apply unitaries $U_{(i,k)}$ that transpose $\ket {e_i}$ and $\ket {e_k}$ and leaves the rest of elements fixed to see that, for any $k\in\{1,...,d\}$,
\begin{equation*}
  \bra{e_i e_i}A\ket{e_i e_i}=\bra{U_{(ik)}\otimes U_{(ik)} e_i e_i} A U_{(ik)}\otimes U_{(ik)}\ket{e_i e_i}=\bra{e_k e_k}A\ket{e_k e_k}.
\end{equation*}
Hence,
\begin{equation*}
\bra{e_i e_i}A\ket{e_i e_i} = \gamma ,\ \forall i,
\end{equation*}
for some $\gamma\in\mathbb{C}$. The same unitaries together with $U_{(jl)}$, defined similarly, gives us
\begin{eqnarray*}
\bra{e_i e_j}A\ket{e_i e_j} = \alpha ,\ \forall i,j,\ i\neq j \\
\bra{e_i e_j}A\ket{e_j e_i} = \beta ,\ \forall i,j,\ i\neq j
\end{eqnarray*}
with $\alpha, \beta \in \mathbb{C}$. The unitaries of the form $U_{(i,k)(j,l)}$ tell us that
\begin{equation*}
  \bra{e_ie_j}A\ket{e_ke_l}=\bra{e_ke_l}A\ket{e_ie_j},
\end{equation*}
and hence all entries are real, that is, actually we have $\alpha, \beta,\gamma\in\mathbb{R}$.

At this point, we already know that $A$ is described by
\begin{equation*}
A\kett {e_i}{e_i}=\gamma\kett {e_i}{e_i}, \ \ \ A\kett {e_i}{e_j} = \alpha \kett {e_i}{e_j} + \beta\kett {e_j}{e_i}.
\end{equation*}

Using unitaries\footnote{This is a variation of the Hadamard operator $H=H_{01}$.} $H_{ij}$ such that
\begin{equation*}
H_{ij}\ket i = \frac{\ket i + \ket j}{\sqrt2},\ \ \
H_{ij}\ket j = \frac{\ket i - \ket j}{\sqrt2}
\end{equation*}
and the rest remains fixed, we see that $\gamma=\alpha+\beta$. Indeed, in the particular case where $d=2$, we have
\begin{equation*}
[A]=\left(
      \begin{array}{cccc}
        \gamma & 0 & 0 & 0 \\
        0 & \alpha & \beta & 0 \\
        0 & \beta & \alpha & 0 \\
        0 & 0 & 0 & \gamma \\
      \end{array}
    \right), \ \ [H \otimes H]=\frac{1}{\sqrt2}\left(
                                            \begin{array}{cccc}
                                              1 & 1 & 1 & 1 \\
                                              1 & -1 & 1 & -1 \\
                                              1 & 1 & -1 & -1 \\
                                              1 & -1 & -1 & 1 \\
                                            \end{array}
                                          \right),
\end{equation*}
thus
\begin{equation*}
[H\otimes H A H^{\dagger}\otimes H^{\dagger}]= \left(
                                               \begin{array}{cccc}
                                                 \frac{\gamma+\alpha+\beta}{2} & 0 & 0 & \frac{\gamma-\alpha-\beta}{2} \\
                                                 0 & \frac{\gamma+\alpha-\beta}{2} & \frac{\gamma-\alpha+\beta}{2} & 0 \\
                                                 0 & \frac{\gamma-\alpha+\beta}{2} & \frac{\gamma+\alpha-\beta}{2} & 0 \\
                                                 \frac{\gamma-\alpha-\beta}{2} & 0 & 0 & \frac{\gamma+\alpha+\beta}{2} \\
                                               \end{array}
                                             \right),
\end{equation*}
which equals $[A]$ if and only if $\gamma=\alpha+\beta$.

Now, we only have to notice that the non-zero matrix elements of the identity have the form $\ketbra{e_j e_k}{e_j e_k}$, for $j,k=1,...,d$ and the non-zero matrix elements of the flip are $\ketbra{e_j e_k}{e_k e_j}$, with $j,k=1,...,d$. So, when we sum a multiple of the identity and a multiple of the flip, the only entries that overlap are $\bra{e_j e_j}\ket{e_j e_j}$.

Since this describes all non-zero matrix elements of $A$, we conclude that $A=\alpha I + \beta V$.\end{proof}

This result is valid for all operators $U \otimes U$-invariants, in particular for Werner states. So, for each of these states, there are parameters $\alpha, \beta$ such that
\begin{equation*}
  W=\alpha I + \beta V.
\end{equation*}
With the constraint that a state has unitary trace and using $\tr(I)=d^2$ and $\tr(V)= d$ (the flip operator permutes the basis elements, standing fixed only $\ket{e_i e_i},\ i=1,...,d$), we are able to bond both parameters together,
\begin{equation*}
  1=\tr(W) = \alpha \tr(I) + \beta \tr(V) = \alpha d^2 + \beta d,
\end{equation*}
and thus
\begin{eqnarray}\label{alfabeta}
\alpha = \frac{(1-\beta d)}{d^2}.
\end{eqnarray}

So we need only one parameter to determine a Werner state. In order to study which Werner states are entangled, we will make use of the parameter $\phi$ defined by
\begin{eqnarray*}
  \phi = \tr(WV) = \tr([\alpha I + \beta V]V) = \alpha \tr(V) + \beta \tr(V^2) \\ = \alpha d + \beta d^2 = \frac{(1-\beta d)}{d} + \beta d^2 = \frac{1-\beta d +\beta d^3 }{d},
\end{eqnarray*}
where we have used that $V^2=I$ and Eq. $(\ref{alfabeta})$. This way we obtain
\begin{equation*}
  \beta=\frac{d\phi -1}{d^3-d}, \ \ \ \alpha=\frac{d-\phi}{d^3-d}
\end{equation*}
and therefore
\begin{equation}\label{werner state}
  W=\frac{(d-\phi)I+(d\phi -1)V}{d^3-d}.
\end{equation}

Thus the task of constructing a local model for Werner states is the task of showing that the integral in $(\ref{local})$ is equal to
\begin{equation}
 \label{probconjunta} \tr(WP_a\otimes Q_b)=\frac{(d-\phi)\tr(P_a)\tr(Q_b)+(d\phi -1)\tr(P_aQ_b)}{d^3-d},
\end{equation}
where we have used the formulas $\tr(R\otimes S)=\tr(R)\tr(S)$ and $\tr(V A\otimes B)=\tr(AB)$ (Propositions \ref{kamus}, \ref{kamuss}).

The next results show how useful this parametrization is to study the entanglement of a Werner state.

\begin{lema}\label{rafael} If $\rho\in\mathcal{D(H)}$, then $\tr(V\rho)\in[-1,1]$, where $V$ is the flip operator.
\end{lema}

\begin{proof} By noticing that $V^{\dagger}=V$, we see that all eigenvalues $\lambda$ of $V$ are real. Since $V^2=I$, we have
\begin{equation*}
V\ket v =\lambda\ket v \implies \ket v=V^2\ket v=\lambda^2\ket v\implies \lambda^2=1,
\end{equation*}
hence $\lambda\in\{\pm 1\}$.

Therefore, the least value achieved by $\tr(V\rho)$ is $\tr(-\rho)=-1$, corresponding to the situation where $\rho=\ketbraa{\psi}$ and $\ket{\psi}$ is an eigenvector of $V$ associated with $-1$. Similarly, the largest value of $\tr(V\rho)$ is 1, obtained when $\rho=\ketbraa{\psi}$ and $\ket{\psi}$ is an eigenvector associated with $1$.
\end{proof}

\begin{teo}\label{ob-ladi ob-lada} The flip operator $V$ is the optimal entanglement witness for Werner states. That is, if $W\in\mathcal{D}(\mathbb{C}^d\otimes\mathbb{C}^d)$ is a Werner state, then it is separable if and only if $\phi=\tr(VW)\geq0$.
\end{teo}

\begin{proof} Note that
\begin{equation*}
\tr(V(\ketbraa{ab}))=\tr((V\ket{ab})\bra{ab})=\tr(\ketbra{ba}{ab})=\delta_{ab},
\end{equation*}
i.e., equals 1 if $\ket a=\ket b$ and 0 otherwise. Since a separable state is a convex combination of product states of the $\ketbraa{ab}$ form, we conclude that
\begin{equation*}
  \phi=\tr(VA)\in[0,1]
\end{equation*}
whenever $A$ is a separable state. This proves the first part of the theorem.

In order to prove the converse, recall that the Werner states are uniquely determined by the paramater $\phi$, which lies in $[-1,1]$ according to Lemma \ref{rafael}. Firstly, we observe that is enough to show that the Werner states $W_0, W_1$, corresponding to $\phi=0$ and $\phi=1$, are separable. Indeed, assuming this, then for every $\phi\in[0,1]$ the corresponding Werner state $W_{\phi}$ is given by
\begin{equation*}
W_{\phi}=(1-\phi)W_0+\phi W_1.
\end{equation*}
This follows from the linearity of the $U\otimes U$ conjugation (which implies the $U\otimes U$-invariance of $W_{\phi}$), the achievement of the $\tr(W_{\phi})=1$ condition and from the calculation
\begin{equation*}
\tr(VW_{\phi})=(1-\phi)\tr(VW_0)+\phi \tr(VW_1)=\phi.
\end{equation*}
Thus we conclude that $W_{\phi}$ is separable, since is a convex combination of separable states.

The separability of $W_0$ and $W_1$ will be proved by using the map
\begin{equation*}
  \mathbb{P}:A\mapsto\int{(U\otimes U)A(U^{\dagger}\otimes U^{\dagger})dU},
\end{equation*}
where $dU$ denotes the Haar measure of the unitary transformations group of $\mathbb{C}^d$. (The Haar measure is the unique nonzero measure which is invariant under the group operation.)

$\mathbb{P}$ (which is known as \emph{twirling}) takes arbitrary density matrices $A\in \mathcal{D}(\mathbb{C}^d\otimes \mathbb{C}^d)$ and returns Werner states. Indeed, due to the invariance of $dU$ and the fact that product of unitaries is also unitary, it follows
\begin{eqnarray*}
  \widetilde{U}\otimes \widetilde{U}\mathbb{P}(A) \widetilde{U}^{\dagger}\otimes \widetilde{U}^{\dagger}&=& \int{\widetilde{U} \otimes \widetilde{U} (U\otimes U)A(U^{\dagger}\otimes U^{\dagger})\widetilde{U}^{\dagger}\otimes \widetilde{U}^{\dagger}dU }\\ &=& \int{(\widetilde{U}U\otimes \widetilde{U}U)A(U^{\dagger}\widetilde{U}^{\dagger}\otimes U^{\dagger}\widetilde{U}^{\dagger})dU}\\ &=&\int{(U\otimes U)A(U^{\dagger}\otimes U^{\dagger})dU}\\
  &=&\mathbb{P}(A),
\end{eqnarray*}
which also proves that $\mathbb{P}$ is a projection (onto the $U\otimes U$-invariant operators subspace). Then $\mathbb{P}(A)$ depends only on the parameter $\tr(V \mathbb{P}(A))$, and we can utilize the $U \otimes U$-invariance of $V$ to show that
\begin{eqnarray*}
  \tr(V \mathbb{P}(A))=\int{\tr(V (U\otimes U)A(U^{\dagger}\otimes U^{\dagger}))dU}=\\
  \int{\tr((U^{\dagger}\otimes U^{\dagger}) V (U\otimes U)A)dU}=\tr (VA).
\end{eqnarray*}
Also, if $A$ is separable, so is each $(U\otimes U)A(U^{\dagger}\otimes U^{\dagger})$ and hence $\mathbb{P}(A)$ is also separable. This means that $\mathbb{P}$ preserves the separability and the trace of the operator times $V$.

Now, let $\rho=\rho_0\otimes\rho_1, \sigma=\sigma_0\otimes\sigma_1 \in\mathcal{D}(\mathbb{C}^d\otimes\mathbb{C}^d)$ be separable states. Taking $\rho_0$ and $\rho_1$ to be orthogonal (e.g.: $\rho_0=\ketbraa0, \rho_1=\ketbraa1$) and making use of Prop. \ref{kamuss}, we have
\begin{equation*}
\tr(V\mathbb{P}(\rho))=\tr(V\rho)=\tr(\rho_0\rho_1)=0.
\end{equation*}
Thus $\mathbb{P}(\rho)$ is precisely the Werner state determined by $\phi=0$, which is $W_0$. Since $\mathbb{P}$ preserves separability, we conclude that $W_0$ is separable.

Similarly, for $\sigma_0=\sigma_1=\ketbraa{\psi}$ we have
\begin{equation*}
\tr(V\mathbb{P}(\sigma))=\tr(V\sigma)=\tr(\ketbraa{\psi})=1.
\end{equation*}
Thus $\mathbb{P}(\sigma)=W_1$, whose separability is implied by the separability of $\sigma$, completing the proof.
\end{proof}

The first part of the proof establishes that the flip operator is an entanglement witness\footnote{It seems fair to say that the concept of entanglement witness only appeared in Ref. \cite{HHH2}, dating from 1996 - hence, posterior to Werner's paper.} for quantum states, i.e., we can calculate $\tr(V\rho)$ for any state $\rho$ and conclude that it is entangled if the result is negative (although, in general, no conclusion can be taken if the trace results to be positive). The second part shows that the flip is the optimal witness for Werner states, since it reveals the entanglement of any of such states.

\section{Bob's response function and the hidden variables}

Now that we have setted the class of states, we focus on the objects that compose the local model: the hidden variables space, the probability distribution upon it and the response functions. The consequences of the $U\otimes U$-invariance of Werner states is a key observation to motivate the choices we are going to make.

Since $W$ is $U\otimes U$-invariant, we have that the left-hand side of Eq. (\ref{local}) is
\begin{equation*}
\tr(W P_a\otimes Q_b)=\tr(U\otimes UWU^{\dagger}\otimes U^{\dagger} P_a\otimes Q_b)
\end{equation*}
\begin{equation*}
= \tr(WU^{\dagger}\otimes U^{\dagger} P_a\otimes Q_bU\otimes U)=\tr(W(U^{\dagger}P_aU)\otimes(U^{\dagger}Q_bU)).
\end{equation*}
Since $(U^{\dagger}P_iU)^2=U^{\dagger}P_iU$ and $\sum_i{U^{\dagger}P_iU}= U^{\dagger}\sum_i{P_i}U=I$, we see that $\{U^{\dagger}P_iU\}$ also defines a projective measurement. Thus the symmetry of the Werner states implies that the probabilities obtained for any pair of measurements in $\{(\{UP_iU^{\dagger}\}, \{UQ_jU^{\dagger}\})\ |\ U\in\mathcal{L(H)}\ \text{unitary}\}$ must be the same. Therefore, the local model to be constructed must satisfy
\begin{equation}\label{zoro}
\int_{\Lambda}{\p_A(a,\lambda,\{P_i\})\p_B(b,\lambda,\{Q_j\})\pi(\lambda)d\lambda}= \int_{\Lambda}{\p_A(a,\lambda,\{UP_iU^{\dagger}\})\p_B(b,\lambda,\{UQ_jU^{\dagger}\})\pi(\lambda)d\lambda}.
\end{equation}
This motivates the setting of the following arrangements, in the construction of the local model.

The space $\Lambda$ of hidden variables will be chosen to be the unit sphere $\{\lambda \in \mathbb{C}^d ; ||\lambda||=1\}$, that is, the hidden variables can be seen as pure quantum states of the local systems. We will also impose a symmetry condition to the response functions $\p_A$ and $\p_B$ used by Alice and Bob. These will be taken to depend on the family $\{P_i\}$ of orthogonal projections, but not on the their labeling or their eigenvalues. Moreover, it suffices to consider only the case where the projectors are one-dimensional, since for projections of higher dimension the response functions can be chosen as a sum of response functions of one-dimensional projections. The symmetry condition to be imposed on the response functions is given by the relation
\begin{equation}\label{simetriaP}
  \p(a,\lambda,\{U^{\dagger}P_iU\})=\p(a, U\lambda,\{P_i\}).
\end{equation}

With this two requirements, Eq. (\ref{zoro}) is equivalent to
\begin{equation*}
\int_{\Lambda}{\p_A(a,\lambda,\{P_i\})\p_B(b,\lambda,\{Q_j\})\pi(\lambda)d\lambda}= \int_{\Lambda}{\p_A(a,U\lambda,\{P_i\})\p_B(b,U\lambda,\{Q_j\})\pi(\lambda)d\lambda},
\end{equation*}
which is satisfied if we set the measure $\pi(\lambda)$ to be the unique measure invariant under unitaries on $\Lambda$.

Hence, to complete our description of the local model, we still are left to decide the response functions $\p_A$ and $\p_B$, respecting Eq. (\ref{simetriaP}). The simplest choice for such functions is the one we will adopt for Bob's system,
\begin{equation}\label{bob_response}
  \p_B(b,\lambda,\{Q_j\})=\tr(Q_b\ket{\lambda}\bra{\lambda})=\bra{\lambda}Q_b\ket{\lambda},
\end{equation}
for the observable $Q=\sum_b{\beta_bQ_b}$. Notice that, by doing this, we are saying that Bob does things ``in the Quantum fashion", since his probabilities will be given by treating the hidden variable as a pure state and using Born's rule to calculate probabilities. Also, the function works as we wish in ($\ref{simetriaP}$),\footnote{For a comment on the abuse of Dirac's notation done here, check Section 1.1.}
\begin{eqnarray*}
  \p_B(b,\lambda,\{U^{\dagger}Q_jU\})=\bra{\lambda}U^{\dagger}Q_bU\ket{\lambda}=
  \bra{U\lambda}Q_b\ket{U\lambda}=\p_B(b,U\lambda,\{Q_j\}).
\end{eqnarray*}

We now proceed to show that using Eqs. ($\ref{simetriaP}$) e ($\ref{bob_response}$), we will be able to reduce the calculation of any of the integrals in the form of Eq. (\ref{local}) to the computation of a single integral.

For any positive integrable function $f:\Lambda \rightarrow \mathbb{R}$, we can consider the positive operator\footnote{There is a misprint in the definition of $\tilde{\rho}$ in the original paper: the $f(\lambda)$ is missing, in the integrand.}
\begin{equation}\label{cassius}
  \tilde{f}=\int_{\Lambda}{\ketbraa{\lambda}f(\lambda)\pi(\lambda)d\lambda}\in\mathcal{L}(\mathbb{C^d}).
\end{equation}
Then, using Bob's response function with $\{Q_j\}$ and $b$ fixed, we have
\begin{equation*}
  \int_{\Lambda}{\p_B(b,\lambda,\{Q_j\})f(\lambda)\pi(\lambda)d\lambda}=\int_{\Lambda}{\tr(Q_b\ketbraa{\lambda})f(\lambda)\pi(\lambda)d\lambda}=
  \tr(Q_b\tilde{f}).
\end{equation*}

In particular, for each fixed $\{P_i\}$ and $a$, we can relate to $\p_A$ the positive operator $\tilde{\p}_A$ given in Eq. (\ref{cassius}), such that
\begin{equation}\label{PtilA}
  \int_{\Lambda}{\p_B(b,\lambda,\{Q_j\}) \p_A(a,\lambda,\{P_i\}) \pi(\lambda) d\lambda} =\tr(\tilde{\p}_A(a,\{P_i\})Q_b).
\end{equation}
Notice that the left-hand side of the above equation equals the right-hand side of Eq. (\ref{local}). Using the $U$-invariance of $\pi$, we find
\begin{eqnarray*}
  \tr(U\tilde{\p}_A(a,\{P_i\})U^{\dagger}Q_b) &=& \tr(\tilde{\p}_A(a,\{P_i\})U^{\dagger}Q_bU) \\ &=& \int_{\Lambda}{\p_A(a,\lambda,\{P_i\})\bra{\lambda}U^{\dagger}Q_bU\ket{\lambda}\pi(\lambda)(d\lambda)} \\
   &=& \int_{\Lambda}{\p_A(a,U\lambda,\{P_i\})\bra{\lambda}Q_b\ket{\lambda}\pi(\lambda)(d\lambda)}  \\
   &=& \int_{\Lambda}{\p_A(a,\lambda,\{UP_iU^{\dagger}\})\bra{\lambda}Q_b\ket{\lambda}\pi(\lambda)(d\lambda)} \\
   &=& \tr(\tilde{\p}_A(a,\{UP_iU^{\dagger}\})Q_b).
\end{eqnarray*}
Since this holds for all one-dimensional projections $Q_b$, we have
\begin{equation*}
U\tilde{\p}_A(a,\{P_i\})U^{\dagger}= \tilde{\p}_A(a,\{UPU^{\dagger}\}).
\end{equation*}
In particular, we see that, if $U$ commutes with every $P_i$, then $U$ commutes with $\tilde \p_A$:
\begin{equation*}
  U\tilde{\p}_A(a,\{P_i\})U^{\dagger}= \tilde{\p}_A(a,\{UP_iU^{\dagger}\})=\tilde{\p}_A(a,\{P_iUU^{\dagger}\})=\tilde{\p}_{A}(a,\{P_i\}).
\end{equation*}
That is to say (check Corollary \ref{tchatchatcha}) that $\tilde \p_A$ (and the same is true for $U$) has a representation
\begin{equation*}
  \tilde{\p}_A(a,\{P_i\})=\sum_i{\omega(a,i)P_i}.
\end{equation*}

Since $\p_A$ is not to depend on the labeling of the projections $P_i$, we conclude that $\omega(a,i)$ depends only on whether $a=i$ or not. Hence
\begin{equation*}
\tilde{\p}_A(a,\{P_i\})=\omega P_a+\omega ' \sum_{i \neq a}P_i=\omega_1 P_a+ \omega_2 I,
\end{equation*}
for some $\omega_1, \omega_2\in \mathbb{R}$, which are independent of $P_i,\ i\neq a$. Since
\begin{eqnarray*}
  \sum_a\tilde{\p}_A(a,\{P_i\})&=&\sum_a{\int_{\Lambda}{\p_A(a,\lambda,\{P_i\})\ket{\lambda}\bra{\lambda}\pi(\lambda)d\lambda}}\\
  &=& \int_{\Lambda}{\sum_a{(\p_A(a,\lambda,\{P_i\})) \ket{\lambda}\bra{\lambda}\pi(\lambda)d\lambda}}\\ &=&\int_{\Lambda}{\ket{\lambda}\bra{\lambda}\pi(\lambda)d\lambda}\\ &=&I
\end{eqnarray*}
we must have
\begin{equation*}
\sum_a{(\omega_1 P_a+ \omega_2I)}=\omega_1I+d\omega_2I=I,
\end{equation*}
thus
\begin{equation}\label{tchururu}
\omega_1+d\omega_2=1.
\end{equation}

Hence for computing $\tilde{\p}_A(a,\{P_i\})$ for a given $a$, it suffices to calculate the value of the expression $\tr(Q_b\tilde{\p}_A(a, \{P_i\}))$ for only one arbitrary choice of $Q_b$. (This is the great advantage of dealing with $\tilde{\p}_A$: it is uniform over all $Q_b$.) Choosing $Q_b=P_a$ \cite{Wer}, we have
\begin{equation}\label{shalala}
  \tr(\tilde{\p}_A(a,\{P_i\})P_a)=\tr([\omega_1 P_a+ \omega_2 I]P_a)=\tr(\omega_1 P_a+ \omega_2 P_a)=\omega_1+\omega_2.
\end{equation}
Assuming in advance that our local model works, i.e., that
\begin{equation*}
\tr(W P_a\otimes Q_b)= \int_{\Lambda}{\p_A(a,\lambda,\{P_i\})\p_B(b,\lambda,\{Q_j\})\pi(\lambda)d\lambda},
\end{equation*}
and using Eq. (\ref{PtilA}), we have
\begin{equation*}
\tr(WP_a\otimes Q_b)= \tr(Q_b\tilde{\p}_A(a,\{P_i\})),
\end{equation*}
hence $\tr(Q_b\tilde{\p}_A(a,\{P_i\}))$ can be easily calculated (once $W, P_a$ and $Q_b$ are determined and well-known) and Eqs. (\ref{tchururu}), (\ref{shalala}) defines explicitly $\omega_1,\omega_2$ and thus $\tilde{\p}_A$.

However, it is not really important the values of $\omega_1, \omega_2$. The main point here is that, for fixed $\{P_i\}$ and $a$, solving the problem for $Q_b=P_a$ implies solving it for any $Q_b$. In this situation, the only probability we need to achieve with our local model is
\begin{eqnarray*}
\tr(WP_a\otimes P_a)&=&\frac{(d-\phi)(\tr P_a)(\tr P_a)+(d\phi -1)\tr(P_a^2)}{d^3-d}\\ &=&\frac{(d-\phi)+(d\phi -1)}{d^3-d}\\ &=&\frac{d(1+\phi)-(1+\phi)}{d^3-d}\\ &=&\frac{(d-1)(1+\phi)}{d(d^2-1)}\\ &=&\frac{1+\phi}{d(d+1)},
\end{eqnarray*}
that is, it suffices to show that
\begin{equation*}
\frac{1+\phi}{d(d+1)}=\int_{\Lambda}{\p_A(a,\lambda,\{P_i\})\bra{\lambda}P_a\ket{\lambda}\pi(\lambda)d\lambda}.
\end{equation*}
But this last equality is trivially achieved by setting
\begin{equation}\label{defphi}
\phi=-1+d(d+1)\int_{\Lambda}{\p_A(a,\lambda,\{P_i\})\bra{\lambda}P_a\ket{\lambda}\pi(\lambda)d\lambda}.
\end{equation}
In other words, we have constructed a local model for the state given by the parameter $\phi$ determined by Eq. (\ref{defphi}). However, we do not have much information about this state. Particularly, we are still under the risk of the integral in Eq. (\ref{defphi}) provides a parameter $\phi$ which determines a separable state, and thus we are only presenting a sophisticated proof of a well known fact to us: that such separable state is local.

Therefore, it remains to show that the right-hand side of Eq. (\ref{defphi}) can be negative, what would gives us a negative $\phi$, corresponding (according with Theorem \ref{ob-ladi ob-lada}) to an entangled state. In order to do that, the card still left up our sleeve is the setting of Alice's response function $\p_A(a,\lambda,\{P_i\})$. Our task now is to determine $\p_A$ satisfying the symmetry (\ref{simetriaP}) for which the integral in (\ref{defphi}) becomes as small as possible (since it is always nonnegative), under the constraints $\p_A(a,\lambda,\{P_i\})\geq 0$ and $\sum_a{\p_A(a,\lambda,\{P_i\})}=1$ for all $\lambda$ and all $\{P_i\}$.

We observe that Eq. (\ref{defphi}) determines $\phi$ (and thus $W$) according with the choice of Alice's measurement $\{P_i\}$ and outcome $a$. So, in principle, it seems that we have constructed a state that depends on the measurement to be carried out on Alice's side. This would be a huge problem, since our goal is, for a given a state, be able to simulate the correlations provided by \emph{any} local measurements. We will see, however, that the choice of Alice's response function and the measure $\pi(\lambda)$ guarantees that the integral in (\ref{defphi}) equals the same value for any $a, \{P_i\}$ under consideration.

\section{Alice's response function and a local model for an entangled state}\label{alice response function}

At this point, we might feel tempted to look for a manner to compute the probabilities of Alice's outcomes similar to Bob's. The response function $\p_B(b,\lambda, \{Q_j\})$ employed by Bob has the characteristic property of depending only on $Q_b$, once $b$ and $\lambda$ are fixed. The remaining measurement operators do not influence $\p_B$. This may appear very natural, since Born's rule got us used to something like this, in the quantum context. Nevertheless, it is vital for Werner's construction that the response function $\p_A$ does not have this property.

Indeed, suppose that, by fixing $\lambda$, we have $\p_A=\p_A(P_a)$ (as it is $\p_B$). Then $\p_A$ is a non-negative, summing to one and additive map on families of mutually orthogonal projections, while response function. Associated to $\p_A$, there is a probability measure $\mu$ such that, if $P_i$ is the projector onto the subspace $E_i\subset\mathbb{C}^d$, then
\begin{equation*}
\p_A(P_i)=\mu(E_i),
\end{equation*}
in the sense of Def. \ref{prob meas on C}. According to Gleason's theorem (and assuming $d>2$), there is a density operator $\rho_{\lambda}$ such that
\begin{equation*}
\p_A(P_a)=\tr(P_a\rho_{\lambda}).
\end{equation*}
Substituting such $\p_A$ as well as the previously defined $\p_B$ in the left hand side of Eq. (\ref{local}), we have
\begin{equation*}
\int_{\Lambda}{\tr(P_a\rho_{\lambda})\tr(Q_b\ket{\lambda}\bra{\lambda})\pi(\lambda)d\lambda}= \int_{\Lambda}{\tr( P_a\otimes Q_b \rho_{\lambda}\otimes \ketbraa{\lambda})\pi(\lambda)d\lambda}.
\end{equation*}
Hence, those response functions give rise to a local model that simulates the mixed state
\begin{equation*}
\int_{\Lambda}{\rho_{\lambda}\otimes\ketbraa{\lambda}\pi(\lambda)d\lambda}.
\end{equation*}
However, the above state is clearly separable and, thus, local \emph{a priori}! (Remember, the objective of our quest is an \emph{entangled} local state.) Therefore, we conclude that the response function $\p_A$ should depend on other projectors of $\{P_i\}$ and not just on $P_a$.

With this in our minds, there is only one more observation to be made before the setting of $\p_A$. Since for every fixed $\lambda$ and $P$ the constraints
\begin{equation*}
\p_A(a,\lambda,\{P_i\})\geq 0,\ \ \ \sum_a{\p_A(a,\lambda,\{P_i\})}=1
\end{equation*}
single out a convex set in $\mathbb{R}^d$, we expect that the smallest values of $\phi$ is attained for response functions taking only the values 0 and 1. This suggests the following choice:
\begin{equation}\label{responseA}
  \p_A(a,\lambda,\{P_i\})=
  \begin{cases}
      1, & \text{if} \ \forall k\neq a, \ \bra{\lambda}P_a\ket{\lambda}<\bra{\lambda}P_k\ket{\lambda}  \\
      0, & \text{if} \ \exists k\neq a;
       \ \bra{\lambda}P_a\ket{\lambda}>\bra{\lambda}P_k\ket{\lambda}.
  \end{cases}
\end{equation}

Note that we have left $\p_A(a,\lambda,\{P_i\})$ unspecified at all points where $\bra{\lambda}P_a\ket{\lambda}$ is the minimum of $\bra{\lambda}P_{k}\ket{\lambda}$ but not the unique one. However, since this set is of measure zero, it will not contribute to the integral (\ref{defphi}) anyway, and we may choose on this subset any measurable function satisfying the constraint.

We have written Eq. (\ref{responseA}) in such a form that the property postulated in Eq. (\ref{simetriaP}) is manifest. Moreover, $\p_A$ is independent of the labeling of the $P_a$ in the sense that it only depends of the set of numbers $\bra{\lambda}P_{k}\ket{\lambda}$, but not on their ordering.

Substituting (\ref{responseA}) in the integrand in (\ref{defphi}), we have
\begin{equation}\label{mestre cristal}
  \int_{\Lambda}{\p_A(a,\lambda,\{P_i\})\bra{\lambda}P_a\ket{\lambda}\pi(\lambda)d\lambda}= \int_{\Lambda_a}{\bra{\lambda}P_a\ket{\lambda}\pi(\lambda)d\lambda},
\end{equation}
where $\Lambda_a=\{\lambda \in \Lambda;\forall k\neq a, \ \bra{\lambda}P_a\ket{\lambda}<\bra{\lambda}P_k\ket{\lambda}\}$. Setting $\bra{\lambda}P_i\ket{\lambda}=u_i$, through the relation
\begin{equation*}
\ket{\lambda}\mapsto (u_1,..., u_d)
\end{equation*}
we can identify the space of pure states of $\mathbb{C}^d\otimes \mathbb{C}^d$ with the simplex of $d$ vertices $\mathcal{S}=\{(u_1,...,u_d)\in \mathbb{R}^d; u_1+...+u_d=1\}$, which is embedded in $\mathbb{R}^{d-1}$, and $\Lambda_a$ with the subset $\mathcal{S}'\subset \mathcal{S}$ delimited by the hyperplanes $u_a=u_1,..., u_a=u_d$. The simplex is best imagined to ``stand" upon the plane $u_a=0$, so that the $u_a$ represents the ``vertical" axis. $\mathcal{S}'$ thus is the convex set formed by the face of $\mathcal{S}$ resting on the plane $u_a=0$ and vertex on the barycenter of $\mathcal{S}$, which is at height $u_a=1/d$.

To illustrate the above reasoning, let's calculate explicitly the integral in the case where $d=3$. The simplex turns out to be the restricted plane $\mathcal{S}=\{(x,y,z)\in\mathbb{R}^3; x+y+z=1,\ x,y,z\geq0 \}$ and the planar set $\mathcal{S}'\subset\mathcal{S}$ is delimited by the lines $x=y,\ 2z=1-x$ and $2z=1-y$.

\begin{figure}[h!]
  \centering
  \includegraphics[width=8cm]{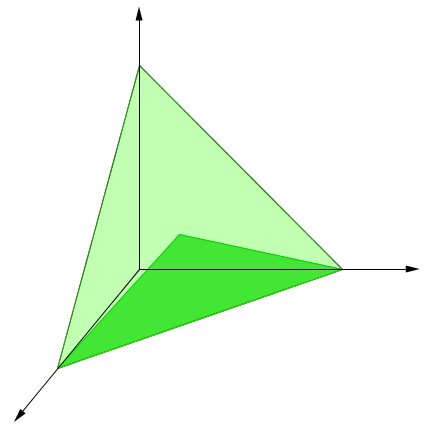}\\
  \caption{\small{The simplex $\mathcal{S}\subset\mathbb{R}^3$ in light color and the subset $\mathcal{S}'$ in dark.}}
\end{figure}

The barycenter of $\mathcal{S}$ (and apex of $\mathcal{S}'$) is the point $(1/3,1/3,1/3)$. We can parameterize the plane containing the simplex using
\begin{equation*}
f(r,s)=(r,s, 1-r-s),\ r,s\in\mathbb{R},
\end{equation*}
for which we have
\begin{equation*}
f_r(r,s)=(1,0,-1),\ f_s(r,s)=(0,1,-1),\ ||f_r\wedge f_s||=\sqrt3.
\end{equation*}

We want to integrate the function $g(r,s)=1-r-s$ over the domain comprehended between the points $(1,0,0),\ (0,1,0)$ and $(0,1/3,1/3)$. Hence, the integral is given by
\begin{eqnarray*}
  &&\int_0^{\frac13}\int_{-2x+1}^{-x+1}(1-r-s)\sqrt3drds + \int_{\frac13}^{1}\int{_{-x/2+1/2}^{-x+1}}{(1-r-s)\sqrt3drds} \\
  &&= \frac{\sqrt3}{2\cdot3^4}+\frac{\sqrt3}{3^4} \\
  &&= \frac{\sqrt3}{2}\frac{1}{3^3}.
\end{eqnarray*}

The normalization (indicated by $\pi(\lambda)$ in Eq. (\ref{mestre cristal})) is done by dividing the resulting expression by $\sqrt3/2$, the area of $\mathcal{S}$, an equilateral triangle of side length $\sqrt2$. We obtain the result $1/3^3$.

In the general case, the integral of $u_a$ over $\mathcal{S}'$ is equal to the height of the barycenter of $\mathcal{S}'$ times the volume of $\mathcal{S}'$. Since the barycenter has height $u_a=1/d^2$ and the volume of $\mathcal{S}'$ is  $1/d$ (the simplex is formed by $d$ pieces congruent to $\mathcal{S}'$ and has volume equal to 1 by definition), we find that the integral results in $1/d^3$.

Substituting this value in Eq. (\ref{defphi}), we have
\begin{equation}\label{luffy}
\phi=-1+\frac{1+d}{d^2}.
\end{equation}
Henceforth, we have constructed a local model for projective measurements for an entangled Werner state, since this value of $\phi$ is negative for all $d\geq2$. We will denote this state by $W_{local}$.

By now, after the considerations over $\p_A$ and the correspondence with the simplex, it is clear that the value of $\phi$ in Eq. (\ref{defphi}) is independent of the choice of $P_a$ and the whole construction results in the state given by $\ref{luffy}$ regardless of our starting point.

\section{Barrett's local model}\label{barrett}

Werner's paper \cite{Wer} was a breakthrough in the foundations of Quantum Theory\footnote{However, it took a while for being acclaimed as so. It is nice to see the history of citations.}, since showed that entanglement and non-locality are distinct features of quantum states. His model is a proof that some aspects of Quantum Theory can be reproduced using only classical resources and that nonlocality is one of the distinctly nonclassical features of Quantum Theory.

\begin{figure}[h!]
  \centering
  \includegraphics[width=15cm]{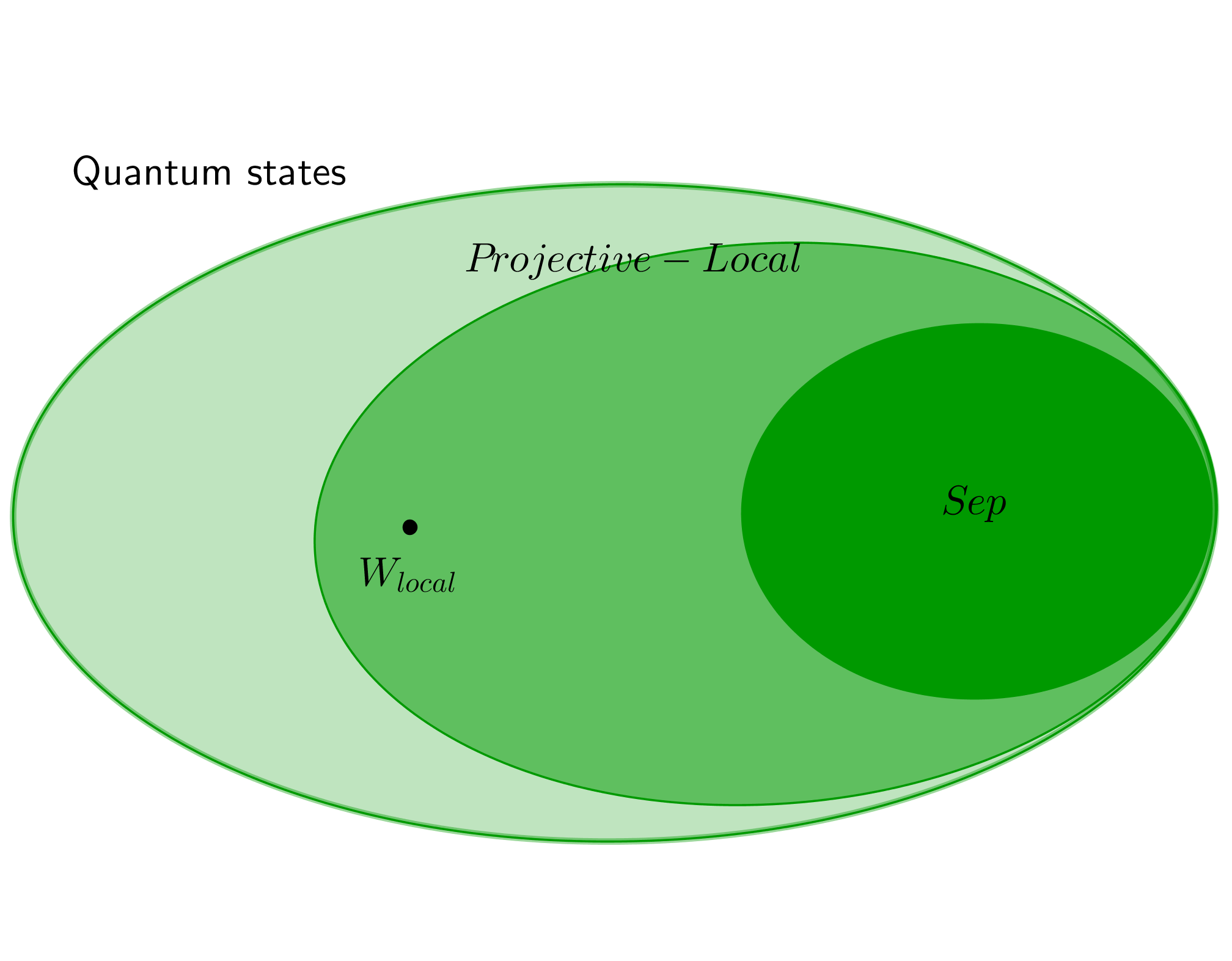}\\
  \caption{\small{Werner's local model shows that the set of separable states is strictly contained in the set of projective-local states: at least $W_{local}$ is projective-local but not separable.}}
\end{figure}

However, the reproduction capability presented by Werner's local model is limited, since it only works for projective measurements. Indeed, quantum mechanics allows us to perform a more general kind of measurement, the POVMs. The natural question raised now is: is there a local model which can reproduce the correlations generated by any POVM applied to an entangled state?

Jonathan Barrett showed that the answer to this question is affirmative \cite{Bar}. Choosing appropriate response functions for Alice and Bob, he was able to create a model that reproduces the correlations given by any POVM over an slightly modified Werner state which is entangled. More explicitly, the state simulated by Barrett is
\begin{equation}\label{shiryu}
  W_{Barrett}=\alpha\frac{2\sum_{i<j;i,j=1}^d{\ket{ij}\bra{ji}}}{d(d-1)}+(1-\alpha)\frac{I_{d\times d}}{d^2},
\end{equation}
where
\begin{equation*}
  \alpha=\frac{(d-1)^{d-1}(3d-1)}{(d+1)d^d}
\end{equation*}
and $d$ is the local dimension. (In Chapter \ref{hidden} we show that the original Werner state also has the form in Eq. (\ref{shiryu}), with $\alpha=(d-1)/d$.) In state $W_{Barrett}$, the first density matrix that appears in the superposition is sometimes called the \emph{projector onto the antisymmetric subspace} of $\mathcal{D}(\mathbb{C}^d\otimes \mathbb{C}^d)$. $W_{Barrett}$ can be shown to be entangled if and only if $\alpha>1/(1+d)$, which happens for all $d\geq2$.

\begin{figure}[h!]
  \centering
  \includegraphics[width=15cm]{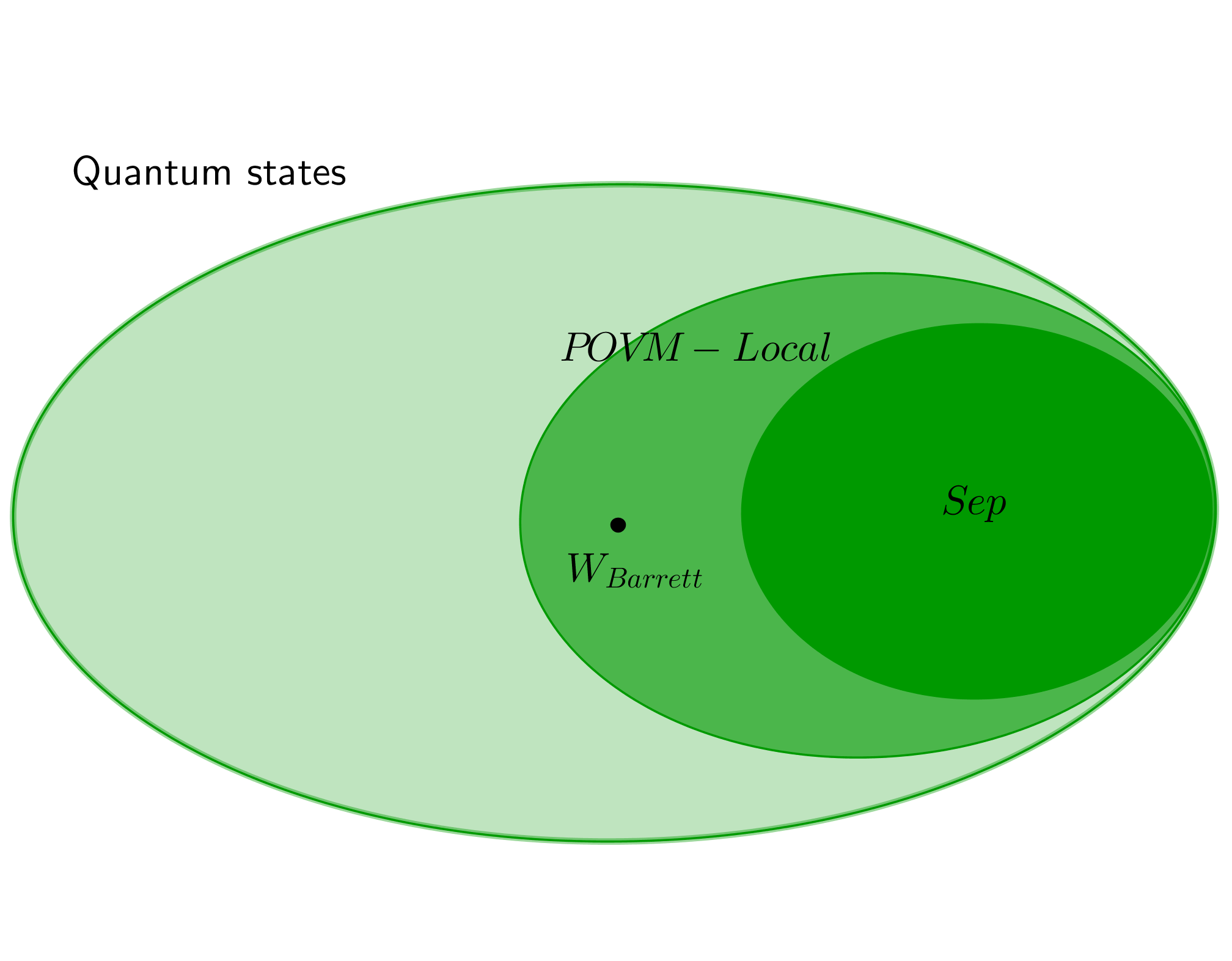}\\
  \caption{\small{Barrett's local model shows that the set of separable states is strictly contained in the set of POVM-local states: at least $W_{Barrett}$ is POVM-local but not separable. It is still an open question if the sets of projective-local and POVM-local states are the same or not; in other words, it is not known if general measurements indeed offer any advantage to detect the nonlocality of projective-local states.}}
\end{figure}

One of the key observations that allowed Barrett to construct the model was that we can restrict our attentions to the case where all the POVM elements are proportional to one-dimensional projectors, that is, POVMs of the form $M=\{M_i\}$ where $M_i=a_iP_i$, with $0\leq x_i \leq 1$ and $P_i$ one-dimensional projector. Though we are not going to detail Barrett's local model, this simple result will be further used in Chapter \ref{genuine}, where another local model for POVM will be detailed, so we prove it now.

\begin{prop}\label{shun} If a local model simulates the correlations of a POVM of the form $M=\{M_i\}$, where $M_i=a_iP_i$, $0\leq a_i \leq 1$ and $P_i$ is a rank-1 projector, then it simulates the correlations of any POVM.
\end{prop}

\begin{proof} Let $N=\{N_j\}$ be an arbitrary POVM. According to the Spectral Decomposition Theorem (Theorem \ref{spectral}), each positive operator $N_j$ can be written as $\sum_k{c_{jk}P_{jk}}$, where $0\leq c_{jk}\leq 1$ and the $P_{jk}$ are one-dimensional operators satisfying $P_{jk}P_{jk'}=\delta_{kk'}P_{jk}$. If each $N_j$ is written in this form, then we can regard as if a more ``detailed" POVM is taking place than the one who is actually been performed, with elements $\{c_{jk}P_{jk}\}$, and our model will make the appropriate predictions. If the outcome referent to $P_{jk}$ is predicted by the model, then we say that the outcome $N_j$ was actually obtained.
\end{proof}

In Barrett's model, the set of hidden variables is the same as the one used by Werner, tridimensional unit vectors $\lambda\in\mathbb{C}^d$, considered with uniform distribution. The response functions defined for the POVMs $M=\{M_i\}=\{x_iP_i\}$ performed by Alice is
\begin{eqnarray*}
  \p_A(i,\lambda,\{M_i\})&=&\bra{\lambda}M_i\ket{\lambda} \chi_{\{x\in\mathbb{R};x\geq0\}}(\bra{\lambda}P_i\ket{\lambda}-1/d) \\
  & & +\left(1-\sum_k{\bra{\lambda}M_k\ket{\lambda} \chi_{\{x\in\mathbb{R};x\geq0\}}(\bra{\lambda}P_k\ket{\lambda}-1/d)} \right)\frac{x_i}{d},
\end{eqnarray*}
where the characteristic function $\chi_S$ is defined to map $s$ to 1 if $s\in S$ and to 0 otherwise. Bob's response function  for $N=\{N_j\}=\{y_jQ_j\}$ is
\begin{equation*}
  \p_B(j,\lambda,\{N_j\})=\frac{1}{d-1}y_j(1-\bra{\lambda}Q_j\ket{\lambda}).
\end{equation*}
More about this subject will be said in Chapter \ref{genuine}.

\chapter{Hidden nonlocality}\label{hidden}

Until Sandu Popescu's work \cite{Pop}, all the treatment of Bell inequalities had a common aspect: they all consider the case in which the system is subjected to a \emph{single} local measurement in each part. Popescu showed that, despite the Werner state referent to the parameter $\phi$ in Eq. (\ref{luffy}) being local for any single measurement, i.e., do not violate any Bell inequality, for a \emph{sequence of measurements}, it does. This property became known as \emph{hidden nonlocality}.

The local state provided by Werner is
\begin{equation}\label{shushu}
  W_{local}=\frac{d+1}{d^3}I_{d\times d}-\frac{1}{d^2}V.
\end{equation}

Defining $\ket{S_{ij}}$ to be ``the singlet state in positions $i$ and $j$", that is,
\begin{equation*}
  \ket{S_{ij}}=\frac{1}{\sqrt2}(\ket{ij}-\ket{ji}),
\end{equation*}
we have
\begin{equation*}
  -2\ket{S_{ij}}\bra{S_{ij}}=\ket{ij}\bra{ji}+\ket{ji}\bra{ij}-\ket{ij}\bra{ij}-\ket{ji}\bra{ji},
\end{equation*}
thus the flip operator $V$ can be written as
\begin{equation*}
  V=I-2\sum_{i<j;i,j=1}^d{\ket{S_{ij}}\bra{S_{ij}}}.
\end{equation*}
Substituting the above relation in Eq. (\ref{shushu}), we can write the Werner state as
\begin{equation}\label{w2}
  W=\frac{1}{d^2}\left(\frac{1}{d}I+2\sum_{i<j;i,j=1}^d{\ket{S_{ij}}\bra{S_{ij}}}\right).
\end{equation}

We will now show that after Alice and Bob perform each one a large rank projective measurement, it is possible that the resulting state violates the CHSH inequality for specified measurements $A, A', B$ and $B'$.

First, each part performs the measurement referent to the projective operator
\begin{equation*}
  P=\ket{1}\bra1 + \ket2 \bra2,
\end{equation*}
that is, the measurement elements involved are $\{P, I-P\}$. The resulting state belongs to the space generated by $\{\ket{1}, \ket{2}\}$ (in which case we will say the outcome was 1) or to the space generated by $\{\ket{3},...,\ket{d} \}$ (in which case we will say the outcome was 0).

Because each part is performing this measurement, we have four possible outcomes: $(0,0), (0,1), (1,0)$ and $(1,1)$. The resulting (unnormalized) state corresponding to the outcome $\{1,1\}$ is
\begin{eqnarray*}
  P\otimes P W_{local} P\otimes P &=& \frac{1}{d^3}P\otimes P I P\otimes P+ \frac{2}{d^2}\sum_{i<j}P\otimes P \ket{S_{ij}}\bra{S_{ij}}P\otimes P\\
   &=& \frac{1}{d^3}P\otimes P + \frac{2}{d^2} \ket{S_{12}}\bra{S_{12}},
\end{eqnarray*}
where we have used the form (\ref{w2}). Let's denote the subspace generated by $\{\ket1, \ket2\}$ by $\mathcal{S}_{12}$. Notice that $P\otimes P := I^{4\times 4}$ is the $4\times 4$ identity matrix acting in the $\mathcal{S}_{12}\otimes \mathcal{S}_{12}$ subspace and zero at the rest. Recalling that $\ket{S_{12}}=\ket{\Psi_-}$, after normalization we obtain
\begin{eqnarray*}
  W' &=& \frac{1}{\tr(\frac{1}{d^3}I^{4\times 4} + \frac{2}{d^2} \ket{\Psi_-}\bra{\Psi_-})}\left(\frac{1}{d^3}I^{4\times 4} + \frac{2}{d^2} \ket{\Psi_-}\bra{\Psi_-} \right)\\
   &=& \frac{1}{(\frac{4}{d^3}+\frac{2}{d^2})d^3} I^{4\times 4} +\frac{2}{(\frac{4}{d^3}+\frac{2}{d^2})d^2} \ket{\Psi_-}\bra{\Psi_-}\\
   &=& \frac{d}{d+2}\left(\frac{1}{2d}I^{4\times 4}+\ket{\Psi_-}\bra{\Psi_-}\right).
\end{eqnarray*}
Notice that as $d$ grows, the state approaches the singlet.

For the second measurement, Alice chooses\footnote{At this point, perhaps should be made a comment on a delicate topic, ``causality". We won't.} randomly between operators $A$ and $A'$ and Bob between $B$ and $B'$. Each of these operators have three different eigenvalues, 1, -1 and 0. The eigenvalues 1 and -1 are non-degenerate and the corresponding eigenvectors belong to the subspace $\mathcal{S}_{12}$. The eigenvalue 0 is highly degenerate and corresponds to the rest of the space, the subspace generated by $\{\ket3, ...,\ket d\}$. The nongenerate part of these operators is chosen such that they yield maximal violation of the CHSH inequality for the singlet state $\ket{\Psi_-}$, that is,
\begin{equation*}
  \bra{\Psi_-}AB+AB'+A'B-A'B'\ket{\Psi_-}=2\sqrt2.
\end{equation*}
In another words, the operators $A, A', B, B'$, when restricted to $\mathcal{S}_{12} \otimes \mathcal{S}_{12}$, are equal to the operators $Q, R, S, T$, respectively, defined in Section \ref{belldesi}.

With these operators and the state $W'$, we have
\begin{equation*}
  \tr(W'[AB+AB'+A'B-A'B'])
\end{equation*}
\begin{equation*}
  =\frac{1}{2d+4}\tr(QS+QT+RS-RT)+\frac{d}{d+2}\bra{\Psi_-}AB+AB'+A'B-A'B'\ket{\Psi_-}.
\end{equation*}
It is not hard to see that $\tr(QS)=\tr(RS)=\tr(RT)=\tr(QT)=0$, hence the first term of the right side of the above equation vanishes and we conclude
\begin{equation*}
  \tr(W'[AB+AB'+A'B-A'B']) = \frac{2\sqrt2d}{d+2}>2
\end{equation*}
for $d>5$.

We conclude that, although Werner's model can simulate all the correlations which arise when only a single measurement is performed on each of the two particles, the model cannot account for the correlations which arise when two consecutive measurements are performed in each particle.

\begin{figure}[h!]
  \centering
  \includegraphics[width=18cm]{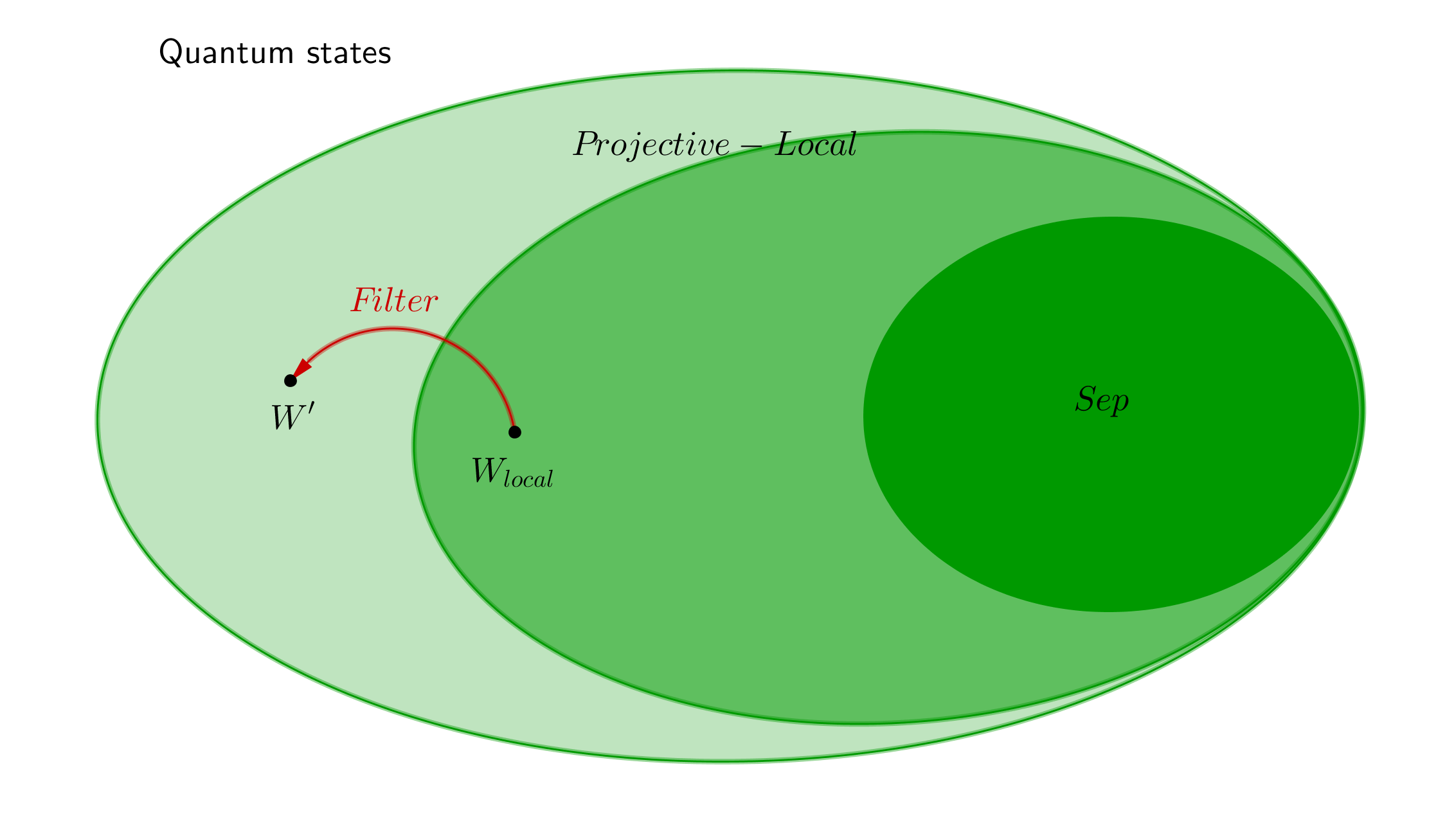}\\
  \caption{\small{The successful application of a local filtering (represented by the arrow in red) on $W_{local}$ results in a state $W'$ which violates CHSH, and thus does not belong to the set of local states.}}
\end{figure}

The main idea here is the introduction of previous local measurements, which are able to take the Werner state $W_{local}$ to a CHSH-violating state with non-zero probability. The nonlocality revelation depends on the success of these measurements. Therefore, they are known as \emph{local filters}: they filter the protocol of revealing nonlocality, in the sense that the process is discarded if they fail, not providing the desired state needed for a Bell inequality violation. In Chapter 7 we will meet such objects again.

\chapter{Gisin-Degorre's local model}\label{degorre}

In 1999, Nicolas Gisin and Bernard Gisin presented a local model for the singlet state based on the detection loophole, that is, the non-maximal efficiency inherent to detectors, in experiments \cite{GG}. The idea was to introduce a ``probability of failure" in the detector of one system (say Alice's), generated by the local hidden variables.

Years later, in 2005, Julien Degorre and co-workers studied a way to quantify nonlocality by measuring the amount of additional resources that should be considered besides the hidden variables in order to simulate the measurements' correlations, namely classical communications, post-selection and nonlocal boxes \cite{DLR}. They considered the Gisin-Gisin model in \cite{GG} and reinterpreted the probability of failure as a probability of rejection by Alice. Then, with only one bit of communication, they were able to reproduce locally the statistics of the outcomes obtained in a EPR experiment (see Section \ref{epr}). However, without any further resources, their resulting protocol happens to simulate the Werner state of $\mathbb{C}^2\otimes \mathbb{C}^2$
\begin{equation*}
  W=\alpha \ket{\psi_-}\bra{\psi_-}+(1-\alpha)\frac{I_{4\times 4}}{4},
\end{equation*}
with parameter $\alpha=1/2$.

In order to agree in notation with the original paper \cite{DLR}, from this point on we are going to use the scalar product symbol for the functional of $\mathbb{R}^3 \times \mathbb{R}^3$, which for unit vectors denotes the co-sine of the angle between them, i.e.,
\begin{equation*}
  v\cdot u:=\cos\theta_{vu},
\end{equation*}
for unit vectors $\ket v, \ket u\in\mathbb{R}^3$.

\section{Shared randomness}

The context of \cite{DLR} is Information Theory, and this influences their approach to the local models problem. The first important observation is that the local hidden variables can be seen as shared randomness, that is, as random variables provided by some source to which both Alice and Bob have access. As discussed in Chapter \ref{nonlocality}, the role of the variables $\lambda$, which are intrinsically random, is to completely explain the correlation between the outcomes of each part. But we can go further and assume that the whole randomness of the system is due to such variables, like if the probabilistic nature of the outcomes is incorporated in the distribution $\pi(\lambda)$ \cite{BCPSW}. In this sense, once we know the hidden variable $\lambda$, the result of every possible measurement is determined and there is no correlation between the systems.

In this view, the local model works by taking a $\lambda$ variable from a random source that works according to $\pi(\lambda)$ and using it to compute the response functions $\p_A, \p_B$. The integral in Eq. (\ref{local}) means that considering all the $\lambda$ that provide outcome $a$ for a measurement $\{M_i\}$, we reach $\p(a|\{M_i\})$. So we can imagine that exists a random source to which all parts have access providing the variables that Alice and Bob input in their response functions. As we will see, it is part of the strategy of the Gisin-Degorre model that Alice sieves such variables before using them, provoking a bias on the initially uniform distribution of them.

Another remarking difference is that the authors are generally interested in the mean quantity of information exchanged in a process, so that it is natural to consider the expected value of an outcome instead of its probability. This means that now we are going to simulate the expectation $\E(A)$, where $A$ is the random outcome of the observable $x\cdot\sigma$, instead of $\p(a|x)$ for some fixed $a$. An immediate consequence is that now our local model will be label-dependent, unlike Werner's.

With the deterministic view of the measurements' results provided by the shared random variables, the local model is now defined by response functions of the form $A(x,\lambda)$ that take values on $\{-1,1\}$, accordingly to the random variable $\lambda\in\Lambda$, which has distribution $\pi(\lambda)$. In this way, the joint expectation of measurements $x\cdot\sigma$ done by Alice and $y\cdot\sigma$ done by Bob becomes
\begin{equation*}
\E(AB)=\int_{\Lambda}{\pi(\lambda)A(x,\lambda)B(y,\lambda)d\lambda}.
\end{equation*}

Now we must also guarantee that the marginals are achieved properly by the response functions, so we must check that
\begin{equation*}
\E(A)=\int_{\Lambda}{A(x,\lambda)\pi(\lambda)d\lambda}
\end{equation*}
and
\begin{equation*}
\E(B)=\int_{\Lambda}{B(y,\lambda)\pi(\lambda)d\lambda}.
\end{equation*}

\section{A local model for the EPR experiment correlations}

In Section \ref{epr}, we presented the EPR experiment, which proves that if Quantum Theory were local, it would be an incomplete theory. There we used the observables $\sigma_x, \sigma_z$, but any other pair of non-commuting observables could replace them. We now present a version slightly modified of the experiment, where Alice and Bob are allowed to measure the spin along any direction they choose, based on Bohm's simplified version.

\begin{definicao} Alice and Bob share a qubit pair in the singlet state $\ket{\Psi_-}=(\ket{01}-\ket{10})/\sqrt2$. Each of them then receive the description of a projective measurement they have to perform in their respective qubit, which can be represented by unit vectors $\ket{x}$ and $\ket{y}$ in $\mathbb{R}^3$. They then obtain measurement outcomes $A,B\in \{1,-1\}$, where 1 corresponds to a spin parallel to the measurement's direction and -1 to spin anti-parallel to this direction.
\end{definicao}

According to Quantum Theory, the outcome of Alice's and Bob's measurements, $A$ and $B$, have the following joint probabilities:
\begin{equation*}
\p(AB)=\frac{1-ABx\cdot y}{4},
\end{equation*}
or, equivalently, their joint and marginal expectation values are given by
\begin{eqnarray*}
  \E(AB) &=& -x\cdot y, \\
  \E(A) &=& 0, \\
  \E(B) &=& 0.
\end{eqnarray*}

So, in order to simulate locally the EPR experiment, we must achieve these three expected values. Because of the Bell inequality violation (Section \ref{belldesi}), we know that this cannot be done using only hidden variables, that is, there is no way to define $\lambda, \pi(\lambda), A(x,\lambda)$ and $B(y,\lambda)$ such that
\begin{equation*}
    \E(AB)=\int{\pi(\lambda)A(x,\lambda)B(y,\lambda)d\lambda},
\end{equation*}
with $\pi(\lambda)$ independent of $x$ and $y$.

However, Degorre and co-workers showed that if we allow the distribution of the hidden variables \emph{to depend of Alice's measurement}, then statistics can be reproduced locally. In order to prove that, we start with two technical lemmas, which will be useful to the next results as well.

\begin{lema}\label{lema1} For any fixed unit vector $x\in\mathbb{R}^3$ we have 
\begin{equation*}
\int_{\mathbb{S}^2}{(x\cdot \lambda)d\lambda}=0.
\end{equation*}
\end{lema}

\begin{proof} Using spherical coordinates and adopting the reference frame where $\ket{x}=(1,0,0)$ and $\ket{\lambda}=(1,\theta,\phi)$, we have
\begin{eqnarray*}
  \int_{\mathbb{S}^2}{(x\cdot \lambda)d\lambda} &=& \int_0^{\pi}{\int_0^{2\pi}{(x\cdot\lambda)\sin\theta d\phi}d\theta}.
\end{eqnarray*}
The scalar product of $\ket x$ and $\ket \lambda$ is the co-sine of the angle between them. With the reference frame we set, by definition this angle is $\theta$, since we put $\ket a$ over the $z$-axis. So we obtain
\begin{eqnarray*}
  \int_0^{\pi}{\int_0^{2\pi}{(x\cdot\lambda)\sin\theta d\phi}d\theta} &=& 2\pi \int_0^{\pi}{\cos\theta\sin\theta d\theta} \\
   &=& 2\pi \left[\frac{1}{2}\sin^2\theta\right]_0^{\pi} \\
   &=& 0.
\end{eqnarray*}

\end{proof}

\begin{lema}\label{lema2} For any fixed unit vector $x\in\mathbb{R}^3$ we have\footnote{This result can be generalized to $\int_{\mathbb{S}^n}{|x\cdot \lambda|d\lambda}=\frac{2}{n}S_{n-1}$, where $S_{n-1}$ denotes the surface area of $\mathbb{S}^{n-1}$ \cite{degorre2006}.} 
\begin{equation*}
\int_{\mathbb{S}^2}{|x\cdot \lambda|d\lambda}=2\pi.
\end{equation*}
\end{lema}

\begin{proof}Under the same assumptions in the proof of Lemma \ref{lema1}, we have
\begin{eqnarray*}
  \int_{\mathbb{S}^2}{|x\cdot \lambda|d\lambda} &=& \int_0^{\pi}{\int_0^{2\pi}{|\cos\theta|\sin\theta d\phi}d\theta} \\
   &=& 2\pi\int_0^{\pi}{|\cos\theta|\sin\theta d\theta} \\
   &=& 2\pi\left(\int_0^{\frac{\pi}{2}}{\cos\theta\sin\theta d\theta}- \int_{\frac{\pi}{2}}^{\pi}{\cos\theta\sin\theta d\theta}\right)\\
   &=& 2\pi\left( \left[\frac{1}{2}\sin^2\theta\right]_0^{\frac{\pi}{2}} - \left[\frac{1}{2}\sin^2\theta\right]_{\frac{\pi}{2}}^{\pi} \right)\\
   &=& 2\pi.
\end{eqnarray*}

\end{proof}

\begin{teo}\label{teosimuEPR} Let $\ket{x}$ and $\ket{y}$ be respectively Alice's and Bob's inputs. If Alice and Bob share a random variable $\ket{\lambda} \in \mathbb{S}^2$ distributed according to a biased distribution with probability density
\begin{equation}\label{bias}
  \pi(\lambda_s|x)=\frac{|x \cdot \lambda_s|}{2\pi}
\end{equation}
then they are able to simulate the correlations of the EPR experiment without any further resource.
\end{teo}

\begin{proof}
First, notice that $\pi$ is indeed a probability density function, since $\pi(\lambda) \geq 0$ for all $\lambda \in \mathbb{S}^2$ and Lemma \ref{lema2} shows that
\begin{equation*}
  \frac{1}{2\pi}\int_{\mathbb{S}^2}{|x \cdot \lambda_s|d\lambda_s}=1.
\end{equation*}

If Alice and Bob set their respective response functions as
\begin{equation}\label{degorre response}
  A(x,\lambda_s)=-\sign(x \cdot \lambda_s),\ B(y,\lambda_s)=\sign( y \cdot \lambda_s),
\end{equation}
where the sign function $\sign(z)$ is defined as $1$ if $z\geq0$ and $-1$ if $z<0$, then the joint expectation value is given by\footnote{Notice that in this case we can also write the sign function $\sign(\ket x \cdot \ket{\lambda_s})$ to be 1 if $\bra{x}\lambda_s\rangle> \bra{x'}\lambda_s\rangle$ and -1 otherwise, much similar to Bob's response function in the Werner model (Chapter \ref{Werner}).}

\begin{eqnarray*}
  \E(AB)&=&\int_{\mathbb{S}^2}{\pi(\lambda_s|x)A(x,\lambda_s)B(y,\lambda_s)d\lambda_s}\\
  &=&\frac{-1}{2\pi}\int_{\mathbb{S}^2}{|x \cdot \lambda_s|\sign( x \cdot \lambda_s)\sign( y \cdot \lambda_s)d\lambda_s}\\
  &=& \frac{-1}{2\pi}\int_{\mathbb{S}^2}{(x \cdot \lambda_s) \sign( y \cdot \lambda_s)d\lambda_s}\\
\end{eqnarray*}

To solve this last integral, we choose a reference frame where $\ket{y}=(0,0,1),\ \ket{x}=(\sin \alpha, 0, \cos \alpha)$ and $\ket{\lambda_s}=(\sqrt{1-\eta^2}\cos \phi, \sqrt{1-\eta^2}\sin \phi, \eta)$, which yields
\begin{eqnarray*}
  \frac{-1}{2\pi}\int_{\mathbb{S}^2}{(x \cdot \lambda_s) \sign( y \cdot \lambda_s)d\lambda_s} &=&
  \frac{-1}{2\pi}\int_{-1}^{1}{\int_0^{2\pi}{[\sin (\alpha) \cos (\phi) \sqrt{1-\eta^2} +\cos (\alpha) \eta]\sign(\eta) d\phi d\eta}} \\
   &=& \frac{-1}{2\pi}\int_{-1}^{1}{\left[\sin\alpha\sqrt{1-\eta^2}\sign(\eta) \int_0^{2\pi}{\cos\phi d\phi}\right.} \\
   && + \left. \int_0^{2\pi}{\cos (\alpha) \eta \sign(\eta) d\phi}\right]d\eta \\
   &=& \frac{-1}{2\pi}\int_{-1}^1{2\pi \cos (\alpha) \eta \sign(\eta) d\eta}\\
   &=& -\left[\int_{-1}^0{\cos (\alpha) \eta d\eta}-\int_0^{-1}{\cos (\alpha)\eta d\eta}\right]\\
\end{eqnarray*}
\begin{eqnarray*}
   &=& -\cos \alpha \left(\left[\frac{\eta^2}{2}\right]_{-1}^0-\left[\frac{\eta^2}{2}\right]_0^{1}\right)\\
   &=& -\cos \alpha\\
   &=& -x \cdot y.
\end{eqnarray*}

Since $\E(A)=0$ is almost exactly the content of Lemma \ref{lema1}, to finish the proof we only have to establish $\E(B)=0$. Let's consider the half-spheres $\mathbb{S}_+, \mathbb{S}_-$ with respect to $x$ and $\lambda_-\in\mathbb{S}_-, \lambda_+\in\mathbb{S}_+$. We have
\begin{eqnarray*}
  \E(B) &=& \int_{\Lambda}{B(y,\lambda_s)\pi(\lambda_s|x)d\lambda_s} \\
   &=& \frac{1}{2\pi}\int_{\Lambda}{|x\cdot\lambda_s|\sign(y\cdot\lambda_s)d\lambda_s} \\
   &=& \frac{1}{2\pi}\left[\int_{\mathbb{S}_+}{(x\cdot\lambda_+)\sign(y\cdot\lambda_+)d\lambda_+} -\int_{\mathbb{S}_-}{(x\cdot\lambda_-)\sign(y\cdot\lambda_-)d\lambda_-}\right].
\end{eqnarray*}
Substituting $\lambda_-=-\lambda_+$, we obtain
\begin{equation*}
\E(B)=\frac{1}{2\pi}\left[\int_{\mathbb{S}_+}{(x\cdot\lambda_+)\sign(y\cdot\lambda_+)d\lambda_+} -\int_{\mathbb{S}_+}{(x\cdot\lambda_+)\sign(y\cdot\lambda_+)d\lambda_+}\right]=0.
\end{equation*}

\end{proof}

The above theorem shows that with a biased distribution of the hidden variables, Alice and Bob are able to do simulate the correlations of the EPR experiment locally. Nevertheless, we can consider the situation where we start with a uniform distribution of the $\lambda$, just like in Werner's model, and then generate the bias by a sampling of the hidden variables performed by Alice.

The process to carry on this task can be divided in two steps. First Alice samples from the uniformly distributed random variables $\lambda$ the biased variable $\lambda_s$ using her knowledge of $a$, with the bias given by Eq. (\ref{bias}). In other words, Alice tests the hidden variable provided by the source (initially unbiased) and accepts it or discards it, in such a way to produce the probability distribution (\ref{bias}). The second step is the communication complexity problem of communicate Bob about which of the hidden variables was accepted. It is to accomplish this second step that we will need additional resources to those allowed by a local model.

\section{The sampling}

To performing the sampling, Alice uses the ``choice method"\footnote{In \cite{DLR}, they initially present the ``rejection method", where Alice rejects $\lambda_k$ if $| a \cdot \lambda_k|<u_k$, where $u_k\sim \text{Unif}[0,1]$. The big drawback is that Alice could reject an arbitrarily large number of samples before accepting one, while in the choice method only one round is enough to perform the sampling. The equivalence of the bias induced by both methods can be seen by noticing that $|a\cdot\lambda|\sim \text{Unif}[0,1]$}, performed in the following way. Once Alice receives $\lambda_0$ and $\lambda_1$ from the random source, she calculates $| a \cdot \lambda_0|$ and $| a \cdot \lambda_1|$. The variable which provides the higher value is accepted as $\lambda_s$.

\begin{proto}(Choice method)

1. Alice picks $\lambda_0 \sim \text{Unif}({\mathbb{S}^2})$

2. Alice picks $\lambda_1 \sim \text{Unif}({\mathbb{S}^2})$

3. If $|x\cdot \lambda_0|>|x\cdot \lambda_1|$, then she accepts $\lambda_0$ and sets $\lambda_s=\lambda_0$; otherwise she accepts $\lambda_1$ and sets $\lambda_s=\lambda_1$.
\end{proto}

\begin{teo}\label{choice}If Alice performs the choice method protocol, then $\lambda_s \sim |x\cdot \lambda_s|/2\pi$ and $\p(\lambda_s=\lambda_1)=\p(\lambda_s=\lambda_2)=1/2$.
\end{teo}

\begin{proof}A uniform distribution of $\lambda$ in $\mathbb{S}^2$ implies a uniform distribution of $|x\cdot \lambda|$ in the interval $[0,1]$, therefore each $\lambda_i$ is accepted with probability $|x\cdot \lambda_i|$. Since the probability density function of the uniform distribution on $\mathbb{S}^2$ is $\frac{1}{2\pi}$, after the protocol we have $\lambda_s \sim |x\cdot \lambda_s|/2\pi$.
\end{proof}

\section{The communication}

Once Alice uses the choice method, Theorem \ref{choice} guarantees that the correct probability distribution (\ref{bias}) occurs. In order to apply Theorem \ref{teosimuEPR} to conclude that the EPR experiment correlations are reproduced, all we need now is to make sure that Bob also knows which is the sampled hidden variable $\lambda_s$, so that he can evaluate the appropriate response function $B(y, \lambda_s)=\sign(y\cdot \lambda_s)$.

Considering only what a standard local hidden variables model allows - that is, shared randomness - the Bell inequality violation guarantees that there is no way to tell Bob about the result of the sampling, since Alice does everything locally. However, in the case where we allow classical communication between both parts, the problem resumes to be trivial. Actually, more: the trivial answer is even optimal. Once the sampling is done, all that Alice has to do is to send one (classical) bit $s\in\{0,1\}$ to Bob saying whether $\lambda_s=\lambda_0$ or $\lambda_s=\lambda_1$.

The problem of ally classical communication to shared randomness with the objective of reproduce quantum correlations has been studied since 1992, when Tim Maudlin showed that, in average, a finite amount of bits would suffice \cite{Mau}. In 2003, Ben Toner and Dave Bacon got the most effective optimization of it, presenting he first local model that, together with one bit, simulate successfully the correlations of the singlet \cite{TB}.

\section{A local model for the Werner state}

With this framework in hands, only a simple observation is enough for us to show that we have a local model (without any further resources) for the $\mathbb{C}^2\otimes \mathbb{C}^2$ Werner state
\begin{equation}\label{werner2x2}
  W_{2\times2}=\alpha \ket{\Psi_-}\bra{\Psi_-}+(1-\alpha)\frac{I_{4\times4}}{4}
\end{equation}
with $\alpha=1/2$.

By linearity, we can see that this state generates the joint expectation value $\E(AB)=-(x\cdot y)/2$.

Giving up the possibility to use classical communication, we come back to the weaker assumption that Alice and Bob have only shared randomness and, therefore, to the situation where Bob is not aware of Alice's sampling, so he cannot distinguish $\lambda_s$ from $\lambda_0$ (or $\lambda_1$). The point is that if Bob always evaluate his respective response function assuming $\lambda_s=\lambda_0$, since $\p(\lambda_s=\lambda_0)= \p(\lambda_s=\lambda_1)=1/2$ on average, he will be right half of the time and thus the correlations will match those of the singlet. In the other half, the independent response functions will generate the probabilities achieved by maximal random noise, that is, by the state $I_{4\times4}/4$.

Thus we are proposing to use the following protocol.

\begin{proto}\label{protodegorre} (Local model for the state $W_{2\times2}$)

1. Alice and Bob share a pair of variables $\lambda_0,\lambda_1\sim \text{Unif}(\mathbb{S}^2)$.

2. Alice performs the choice method applied to $\lambda_0, \lambda_1$ and outputs $A=-\sign(x\cdot \lambda_s)$.

3. Bob outputs $B=\sign(y\cdot \lambda_0)$.

\end{proto}

We now proceed to show that the local model indeed reproduces the correlations of state $W_{2\times2}$.

\begin{lema}\label{schwarz} Alice's response function (\ref{degorre response}) can be rewritten as
\begin{equation*}
A(x,\lambda_s)=-\sign(x\cdot (\lambda_0+\lambda_1)).
\end{equation*}
\end{lema}

\begin{proof} Notice that the sign function allows us to write
\begin{equation*}
-\sign(x\cdot \lambda)=(-1)^{\chi_{S_x}(\lambda)},
\end{equation*}
where $\chi_S$ is the characteristic function of the set $S$ (i.e., $\chi_S(s)$ equals 1 if $s\in S$ and 0 otherwise) and $S_x$ denotes the semi-sphere $\{\lambda \in \mathbb{S}^2; x\cdot \lambda\geq0\}$.

From the definition of $\lambda_s$ in the choice method protocol, we conclude that
\begin{equation*}
\chi_{S_x}(\lambda_0+\lambda_1)=\chi_{S_x}(\lambda_s),
\end{equation*}
which means that $-\sign(x\cdot \lambda_s)=-\sign(x\cdot (\lambda_0+\lambda_1))$.
\end{proof}

\begin{teo} There exists a local model for the Werner state
\begin{equation*}
W_{2\times2}=\frac12 \ket{\Psi_-}\bra{\Psi_-}+\frac12\frac{I_{4\times4}}{4}.
\end{equation*}
\end{teo}

\begin{proof} As already stated, the local model consists of Alice performing the choice method and Bob always assuming that $\lambda_s=\lambda_0$. Using Lemma \ref{schwarz}, Lemma \ref{lema1} and adjusting the differential, we obtain
\begin{eqnarray*}
  \E(AB) &=& -\frac{1}{8\pi^2}\int_{\mathbb{S}^2}{d\lambda_0\int_{\mathbb{S}^2} {d\lambda_1(x\cdot(\lambda_0+\lambda_1))\sign(y\cdot\lambda_0)}} \\
   &=& -\frac{1}{8\pi^2}2\pi (x\cdot y) \int_{\mathbb{S}^2}{d\lambda_1} \\
   &=& -\frac{x \cdot y}{2}.
\end{eqnarray*}

As we also have $\E(A)=\E(B)=0$, the model reproduces the correlations exhibited by $W_{2\times2}$.
\end{proof}

\chapter{Genuine hidden nonlocality}\label{genuine}

So far, this is the road we have been following: Werner constructed a local model for projective measurements for the state $W_{local}$; Barrett constructed a local model for POVMs for the state $\hat{W}$, similar to $W_{local}$ but not the same; and Popescu showed that states $W_{local}$ of dimension $d\geq5$ have hidden nonlocality revealed by a proper sequence of measurements. Since it is not known if $W_{local}$ is local for POVMs or not, these facts are not enough to assure that a sequence of measurements is really necessary to activate its nonlocality. If we adopt the natural hierarchy where a single POVM has a lower ``cost" than a couple of projective measurements, then is reasonable to say that Popescu's activation protocol still left room for optimization. Are two measurements indeed necessary?

On the other hand, it is not known if Barrett's state $\hat{W}$ does not also present hidden nonlocality; so far no-one was able to find a sequence of measurements to display it, but this does not mean that it is impossible to exist one. The natural question raised here is: is there an entangled state, the nonlocality of which can be observed only if sequential measurements are used?

In the paper ``\emph{Genuine hidden quantum nonlocality}" \cite{HQBB}, which is the subject of this chapter, the term ``Genuine" holds in the above sense. Brunner and co-workers presented a state which counts with a local model for POVMs, which nevertheless can be shown to violate the CHSH inequality after a sequence of judiciously chosen local measurements are performed. It was the first and, until present date, the only example of a genuine necessity of more than one measurement.

\section{A local model for dichotomic projective measurements}

We start by the construction of a local model for dichotomic projective measurements for a state of $\mathbb{C}^2\otimes\mathbb{C}^2$. More specifically, we consider the class of states
\begin{eqnarray}\label{ikki}
  \rho_G(q)&=&q\ket{\Psi_-}\bra{\Psi_-}+(1-q)\ket0\bra0\otimes \frac{I_{2\times2}}{2}\\
  \nonumber &=& q\ket{\Psi_-}\bra{\Psi_-}+\frac{1-q}{2}(\ket{00}\bra{00}+\ket{01}\bra{01}),
\end{eqnarray}
recalling that $\ket{\Psi_-}=(\ket{01}-\ket{10})/\sqrt2$ is the singlet state. We will see that state (\ref{ikki}) admits a local model mentioned above when $q\leq1/2$. Nevertheless, we can use the flip operator as witness, as we did with the Werner states in Theorem \ref{ob-ladi ob-lada}, to prove that the state is entangled for\footnote{Actually, the state is entangled for all $q>0$, as can be seen using the Peres-Horodecki Criterion \cite{NC}.} $q>1/3$:
\begin{eqnarray*}
  \tr(V\rho_G) &=& q\tr(V\ket{\Psi_-}\bra{\Psi_-})+\frac{1-q}{2}(\tr(V\ket{00}\bra{00})+\tr(V\ket{01}\bra{01})) \\
   &=& \frac q2 (\tr(V\ketbraa{01})+\tr(V\ketbraa{10})-\tr(V\ketbra{10}{01})-\tr(V\ketbra{01}{10})+ \frac{1-q}{2} \\
   &=& -q+\frac{1-q}{2},
\end{eqnarray*}
which is negative if and only if $q>1/3$.

As in Chapter \ref{degorre}, the statistics we wish to simulate for $\rho_G$ are $\E(A), \E(B)$ and $\E(AB)$. Notice that, by linearity,
\begin{equation*}
\E_{\rho_G}(A)=q\E_{\ket{\Psi_-}\bra{\Psi_-}}(A)+(1-q)E_{\ket0\bra0\otimes\frac{I}{2}}(A)
\end{equation*}
and analogously for $\E(B)$ and $\E(AB)$, where the indices say which state should be considered for each expectation. In Chapter \ref{degorre}, we saw that $\E_{\ket{\Psi_-}\bra{\Psi_-}}(A)=\E_{\ket{\Psi_-}\bra{\Psi_-}}(B)=0$, thus $\E(A)=(1-q)\E_{\ket0\bra0}(A)$ and $\E(B)=(1-q)\E_{\frac{I}{2}}(B)$. Therefore, we have
\begin{eqnarray*}
  \E(A) &=& (1-q)\left[\frac{1+\bra0x\cdot\sigma\ket0}{2}-\frac{1-\bra0x\cdot\sigma\ket0}{2}\right] \\
   &=& (1-q)\bra0x\cdot\sigma\ket0 \\
   &=& (1-q)x_z
\end{eqnarray*}
(where $\ket x=(x_x, x_y, x_z)$) and
\begin{eqnarray*}
  \E(B) &=& (1-q)\left[\tr\left(\frac{I}{2}\frac{I+y\cdot\sigma}{2}\right)- \tr\left(\frac{I}{2}\frac{I-y\cdot\sigma}{2}\right)\right] \\
   &=& \frac{(1-q)}{2}\tr(y\cdot\sigma) \\
   &=& 0.
\end{eqnarray*}

For $\E(AB)$, since $\ket0\bra0\otimes\frac{I}{2}$ is a product state (thus the outcomes of any local measurements are uncorrelated), we have
\begin{eqnarray*}
\E(AB)&=&q\E_{\ket{\Psi_-}\bra{\Psi_-}}(AB)+(1-q)\E_{\ket0\bra0\otimes\frac{I}{2}}(AB)\\
&=& q\E_{\ket{\Psi_-}\bra{\Psi_-}}(AB)+(1-q)\E_{\ket0\bra0}(A)\E_{\frac{I}{2}}(B)\\
&=& q\E_{\ketbraa{\Psi_-}}(AB)\\
&=& -q(x\cdot y),
\end{eqnarray*}
where the last equality was shown in Chapter \ref{degorre}.

The performance of a dichotomic projective measurement is precisely what is done in the EPR experiment (see Section \ref{epr}), the only difference is the state being shared: now we are using the modified singlet state $\rho_G(q)$, instead of the singlet $\rho_G(1)$ itself . Alice (and Bob) will receive as input a vector $\ket{x}$ (and $\ket{y}$), and should simulate the statistics of measuring the qubit observables $x\cdot \sigma$ and $y\cdot \sigma$, with possible outcomes $A,B\in\{-1,1\}$. The following protocol is strongly inspired in Protocol \ref{protodegorre}.

\begin{proto}[Simulation of $\rho_G(1/2)$]\label{proto geneva}
\begin{enumerate}
  \item Alice and Bob share a tridimensional unit vector $\ket{\lambda}$, uniformly distributed on the sphere.
  \item Upon receiving $x$, Alice tests $\lambda$. With probability $|x\cdot\lambda|$, she accepts $\lambda$ and outputs $A=-\sign(x\cdot\lambda)$; otherwise, she outputs $A=\pm1$ with probability $(1\pm\bra0x\cdot\sigma\ket0)/2$.
  \item Bob outputs $B=\sign(y\cdot\lambda)$.
\end{enumerate}
\end{proto}

\begin{teo} There exists a local model that simulates the correlations exhibited by state $\rho_G(q)$ upon the measurement of a dichotomic projective measurement, for any $q\in[0,1/2]$.
\end{teo}

\begin{proof}Let $x\cdot\sigma, y\cdot\sigma$ be the observables measured by Alice and Bob, respectively. After Protocol \ref{proto geneva} is performed, as shown in Theorem \ref{choice}, if Alice accepts $\lambda$, which occurs on average with probability $1/2$ (independently of $x$), $\lambda$ is distributed according to the density $\pi(\lambda)=|x\cdot\lambda|/2\pi$. In this case, since Bob outputs $B=\sign(y\cdot\lambda)$, Theorem \ref{teosimuEPR} says that the correlation between Alice's and Bob's outcomes is
\begin{eqnarray*}
  \E(AB) &=& \frac{-1}{2\pi}\int_{\mathbb{S}_2}{|x \cdot \lambda|\sign(x \cdot \lambda)\sign( y \cdot \lambda)}\\
  &=& -x \cdot y.
\end{eqnarray*}
As the marginal expectations are uniform, i.e., $\E(A)=\E(B)=0$, we recover the singlet correlations.

If Alice rejects $\lambda$, she simulates the statistics of state $\ket0$, while Bob's outcome is uncorrelated and uniformly distributed. Hence the model reproduces exactly the statistics of state (\ref{ikki}) for $q=1/2$.

For $0<q<1/2$, it is sufficient to observe that for $q=0$ the state is clearly local and that the set of all local states is convex. But since we haven't discussed such topics in this text (we recommend for the interested reader \cite{BCPSW}, or \cite{Mur} for a gentler introduction), we will explicitly show such convexity for this class of states, presenting a simple protocol of simulation.

Fix $p\in[0,1]$. The following protocol simulates the state $\breve{\rho}(p)=p\rho_{G}(1/2)+ (1-p)\ketbra{0}{0}\otimes\frac{I}{2}$. (Notice that $\ketbra{0}{0}\otimes\frac{I}{2}=\rho_G(0)$.)

\begin{proto}[Simulation of $\breve{\rho}(p)$]\label{hyoga}

\begin{enumerate}
  \item Alice and Bob share a real number $r$ uniformly distributed on the interval $[0,1]$ and a tridimensional unit vector $\ket{\lambda}$, uniformly distributed on the sphere.
  \item With probability $x$, Alice accepts $r$ and executes Protocol \ref{proto geneva}; otherwise she outputs $A=\pm1$ with probability $(1\pm\bra0x\cdot\sigma\ket0)/2$.
  \item Bob outputs $B=\sign(y\cdot\lambda)$.
\end{enumerate}
\end{proto}

Now, one just have to check that
\begin{eqnarray*}
  \breve{\rho}(p) &=& p\rho_{G}(1/2)+ (1-p)\ketbra{0}{0}\frac{I}{2} \\
  &=& \frac{p}{2}\ketbraa{\psi_-}+\left(\frac{p}{2}+(1-p)\right)\ketbra{0}{0}\frac{I}{2} \\
  &=& \frac{p}{2}\ketbraa{\psi_-}+(1-\frac{p}{2})\ketbra{0}{0}\frac{I}{2} \\
  &=& \rho_G(p/2).
\end{eqnarray*}
Therefore, choosing the proper $p$ on $[0,1]$, with Protocol \ref{hyoga} we can simulate $\rho_G(q)$ for any $q\in[0,1/2]$.
\end{proof}

\section{Revealing nonlocality}

Similar to what we have done in Chapter \ref{hidden}, we are now going to show that after local filtering, the state $\rho_G$ violates the CHSH inequality (see Section \ref{belldesi}).

More specifically, we are going to apply filters of the form
\begin{equation}\label{ayoros}
  F_A=\epsilon\ket0\bra0+\ket1\bra1, \ F_B=\delta\ket0\bra0+\ket1\bra1
\end{equation}
with $\delta=\epsilon/\sqrt q$ to $\rho_G$. In another words, we will perform the measurements $\{F_A, I-F_A\}$ in Alice's side and $\{F_B, I-F_B\}$ in Bob's, and discard the protocol in case that the outcome is referent to operator $I-F_A$ or $I-F_B$. This means that, after filtering, the resulting (unnormalized) state is
\begin{eqnarray*}
\tilde{\rho}_G&=&F_A\otimes F_B \rho_G F_A^{\dagger}\otimes F_B^{\dagger}\\
&=& qF_A\otimes F_B \ket{\Psi_-}\bra{\Psi_-} F_A\otimes F_B + (1-q)(F_A\ket0\bra0F_A)\otimes (F_B\frac{I}{2}F_B)\\
&=& \frac{q}{2}\left[\epsilon^2\ket{01}\bra{01} + \frac{\epsilon^2}{q}\ket{10}\bra{10}-\frac{\epsilon^2}{\sqrt q}(\ket{01}\bra{10}+\ket{10}\bra{01})\right]\\
& & + \frac{1-q}{2}\left[\frac{\epsilon^4}{q}\ket{00}\bra{00} + \epsilon^2\ket{01}\bra{01}\right]\\
&=& \epsilon^2\left[\frac{1}{2}\ket{10}\bra{10}-\frac{\sqrt q}{2}(\ket{01}\bra{10}+\ket{10}\bra{01}) + \frac{1}{2}\ket{01}\bra{01} \right] + O(\epsilon^4)\\
&\simeq& \frac{1}{2}\ket{10}\bra{10}+\frac{1}{2}\ket{01}\bra{01}- \frac{\sqrt q}{2}(\ket{01}\bra{10}+\ket{10}\bra{01})+O(\epsilon^2)
\end{eqnarray*}
By adding and subtracting to both sides $\sqrt{q}(\ket{01}\bra{01}+\ket{10}\bra{10})/2$, we achieve
\begin{equation*}
  \tilde{\rho}_G\simeq \sqrt{q}\ket{\Psi_-}\bra{\Psi_-}+(1-\sqrt q)\frac{\ket{01}\bra{01}+ \ket{10}\bra{10}}{2}+O(\epsilon^2).
\end{equation*}

According to the Horodecki criterion (see Section \ref{horocrit}), we can calculate that $\tilde{\rho}_G$ violates CHSH up to $2\sqrt{1+q}$ (for $\epsilon\rightarrow0$). Thus, state $\rho_G(q)$ is local for all $0\leq q \leq 1/2$ and exhibits hidden nonlocality for projective measurements for all $q>0$. In the unique other know example of hidden nonlocality (Chapter \ref{hidden}), the local dimension was $d\geq5$, making this case (with local dimension $d=2$) the simplest example of the phenomenon.

\section{A local model for POVMs}

Nevertheless, our main goal was not achieved yet. We cannot guarantee that such nonlocality is genuine, in the sense expressed in the beginning of the chapter: the local model constructed accounts only for projective measurements. In principle, a Bell violation can be obtained using POVMs.

However, we now proceed to present a protocol for the construction of a state which actually admits a local model for POVMs, based on another one which is local for projective measurements. To be more precise, once the local model for projective measurements is done for the initial state, we will show that a repeated usage of it provides the simulation of the POVMs for the second one. In this sense, it represents an optimization of utilization of the hidden variables and the response functions of the initial model.

Admitting that a given state $\rho_0$ of local dimension $d$ is local for projective dichotomic measurements, consider the state
\begin{equation}\label{seyia}
  \rho'=\frac{1}{d^2}[\rho_0+(d-1)(\rho_A\otimes\sigma_B + \sigma_A\otimes\rho_B)+ (d-1)^2\sigma_A\otimes\sigma_B],
\end{equation}
where $\sigma_A, \sigma_B$ are arbitrary $d$-dimensional states and $\rho_A=\tr_B(\rho_0), \rho_B=\tr_A(\rho_0)$. ($\rho'$ is indeed a state since is convex combination of states.)

Suppose that Alice and Bob receives as input the POVMs $M=\{M_a\}$ and $N=\{N_b\}$, respectively. Then the expected value $\p(ab)$ which we are interested to reproduce locally is given by
\begin{eqnarray}\label{marin}
  \tr(M_a\otimes N_b \rho') &=& \frac{1}{d^2}[\tr(M_a\otimes N_b \rho_0) + (d-1)\tr(M_a\rho_A)\tr(N_b\sigma_B)\\ &&+ (d-1)\tr(M_a\sigma_A)\tr(N_b\rho_B) + (d-1)^2\tr(M_a\sigma_A)\tr(M_b\sigma_B)].
\end{eqnarray}

Following Proposition \ref{shun}, we can assume that the elements of both POVMs are proportional to one-dimensional projectors, i.e., $M_a=\alpha_aP_a$ and $N_b=\beta_bQ_b$, with $\alpha_a, \beta_b \geq 0$. Note that by normalization of the POVM,
\begin{equation*}
  I=\sum_a{M_a}=\sum_a{\alpha_aP_a},
\end{equation*}
which implies
\begin{equation*}
  d=\tr(I)=\sum_a{\alpha_a\tr(P_a)=\sum_a{\alpha_a}}.
\end{equation*}
Similarly, we find $\sum_b{\beta_b}=d$.

We will show that $\rho'$ is local for POVMs through the following protocol. The protocol is written for Alice, but Bob follows the same procedure.

\begin{proto}\label{proto geneva 2}

\begin{enumerate}
  \item Alice chooses projector $P_a$ with probability $\alpha_a/d$ (notice that $\sum_a{\alpha_a/d}=1$).
  \item She simulates the dichotomic projective measurement $\{P_a, I-P_a\}$ on state $\rho_0$.
  \item If the output of the simulation corresponds to $P_a$, she outputs $a$.
  \item Otherwise, she outputs (any) $a$ with probability $\tr(M_a\sigma_A)$.
\end{enumerate}
\end{proto}

\begin{teo}\label{sheena} If there is a local model for dichotomic measurements over $\rho_0$, then there is a local model that simulates the correlations exhibited by state $\rho'$ given in Eq. (\ref{seyia}) upon the measurement of local POVMs.
\end{teo}
\begin{proof} Suppose Alice and Bob are able to simulate locally the correlations of projective dichotomic measurements for the state $\rho_0$. Let the Protocol \ref{proto geneva 2} be performed and fix $a,b$ of the set of possible outcomes. Our goal is to show that $\p(a,b)=\tr(M_a\otimes N_b\rho')$, where the left side probability is calculated according to the protocol.

First, notice that the probability that Alice (and the same holds for Bob) outputs in step 3 (any outcome) is
\begin{equation*}
  \sum_a{\frac{\alpha_a}{d}\tr(P_a\rho_A)}=\frac{1}{d}\sum_a{\tr(M_a\rho_A)}=\frac{1}{d}.
\end{equation*}

Since each part can output in step (3) or in step (4), four possibilities to obtain outputs $a$ and $b$ arise.

\begin{itemize}
  \item Both Alice and Bob output in step 3, which occurs with probability
  \begin{equation*}
    \frac{\alpha_a}{d}\frac{\beta_b}{d}\tr(P_a\otimes Q_b \rho_0)=\frac{1}{d^2}\tr(M_a\otimes M_b \rho_0);
  \end{equation*}
  \item Alice outputs in step 3 and Bob in step 4. Since Bob outputs in step 4 with probability $(d-1)/d$, this occurs with probability
  \begin{equation*}
    \frac{\alpha_a}{d}\tr(P_a\rho_A)\frac{d-1}{d}\tr(N_b\sigma_B)= \frac{d-1}{d^2}\tr(M_a\rho_A)\tr(N_b\sigma_B);
  \end{equation*}
  \item Bob outputs in step 3 and Alice in step 4, which occurs with probability
  \begin{equation*}
    \frac{d-1}{d^2}\tr(N_b\rho_B)\tr(M_a\sigma_A);
  \end{equation*}
  \item Both output in step 4, which occurs with probability
  \begin{equation*}
    \frac{(d-1)^2}{d^2}tr(M_a\sigma_A)\tr(N_b\sigma_B).
  \end{equation*}
\end{itemize}

Altogether, we find that $\p(a,b)$ matches accurately Eq. (\ref{marin}).

\end{proof}

\section{Revealing genuine nonlocality}

Theorem \ref{sheena} says that the same local model that reproduces the correlations of dichotomic projective measurements for $\rho_0$ can reproduce the correlations of a POVM applied to $\rho'$, if the right protocol is executed. We are now going to use this result to compile everything we saw in the last sections and construct a local state for POVM which violates the CHSH inequality after filtering, proving that a sequence of measurements is indeed necessary in certain cases.

In Eq. (\ref{seyia}), if we set the local dimension to be $d=2$ and $\rho_0=\rho_G(q)$ (given in Eq. (\ref{ikki})), which is local for projective measurements for $q\leq1/2$, we will obtain the reduced states
\begin{eqnarray*}
  \rho_{A}(q)&=&q\frac{I_{2\times2}}{2}+(1-q)\ketbraa0,\\
  \rho_{B}(q)&=&q\frac{I_{2\times2}}{2}+(1-q)\frac{I_{2\times2}}{2}= \frac{I_{2\times2}}{2}.
\end{eqnarray*}
Setting also $\sigma_A=\sigma_B=\ketbraa0$, we define the state
\begin{eqnarray*}
 \rho_G'(q)&=&\frac{1}{4}\left[q\ketbraa{\Psi_-} + (1-q)\ketbraa0\otimes\frac{I_{2\times2}}{2} + q\frac{I_{2\times2}}{2}\otimes\ketbraa0+(1-q)\ketbraa{00} \right.\\&&\left.+ \ketbraa0\otimes\frac{I_{2\times2}}{2}  + \ketbraa{00} \right] \\
 &=& \frac{1}{4}\left[q\ketbraa{\Psi_-}+(2-q)\ketbraa0\otimes\frac{I_{2\times2}}{2} +q\frac{I_{2\times2}}{2}\otimes\ketbraa0 +(2-q)\ketbraa{00}\right].
\end{eqnarray*}

We know that, by construction, $\rho_G'(q)$ with $q\leq1/2$ is local, concerning to general measurements. We are now going to show that, despite of that, it violates the CHSH inequality after filtering, therefore exhibiting genuine hidden nonlocality.

Applying filters of the form (\ref{ayoros}) with $\delta=\epsilon/\sqrt q$ to state $\rho_G'$, we obtain the unnormalized state

\begin{eqnarray*}
  \tilde{\rho}_G' &=& F_A\otimes F_B \rho_G' F_A\otimes F_B\\
   &=& \frac{1}{4}\left[ qF_A\otimes F_B\ketbraa{\Psi_-}F_A\otimes F_B + (2-q)(F_A\ketbraa0 F_A)\otimes(F_B\frac{I}{2}F_B)\right. \\
   && \left. +q(F_A\frac{I}{2}F_A)\otimes(F_B\ketbraa0 F_B)+ (2-q)(F_A\ketbraa0F_A)\otimes(F_B\ketbraa0F_B)  \right] \\
\end{eqnarray*}
\begin{eqnarray*}
   &=& \frac{1}{4}\left[ \frac{q}{2}[\epsilon^2\ketbraa{01}+\frac{\epsilon^2}{q}\ketbraa{10}- \frac{\epsilon^2}{\sqrt q}(\ket{01}\bra{10}+\ket{10}\bra{01})] \right. \\
   && \left. +\frac{2-q}{2}\left(\frac{\epsilon^4}{q}\ketbraa{00}+\epsilon^2\ketbraa{01}\right)\right. \\
   && \left. +\frac{q}{2}\left(\frac{\epsilon^4}{q}\ketbraa{00} +\frac{\epsilon^2}{q}\ketbraa{10} +\frac{2-q}{2}\frac{\epsilon^4}{q}\ketbraa{00}\right) \right]\\
   &=& \frac{\epsilon^2}{4}\left[ \frac{q}{2}\ketbraa{01}+\frac{1}{2}\ketbraa{10} -\frac{\sqrt q}{2}(\ket{01}\bra{10}+\ket{10}\bra{01}) + \ketbraa{01} \right. \\
   && \left. -\frac{q}{2}\ketbraa{01} + \frac{1}{2}\ketbraa{10} \right]+O(\epsilon^4)\\
   &\simeq& \frac{1}{2}\left[ \ketbraa{01} + \ketbraa{10} -\frac{\sqrt q}{2} (\ket{01}\bra{10}+\ket{10}\bra{01}) \right]+O(\epsilon^2)\\
   &=&\frac{\ketbraa{01}+\ketbraa{10}}{2}-\frac{\sqrt q}{2}\frac{\ket{01}\bra{10}+\ket{10}\bra{01}}{2}+O(\epsilon^2)\\
   &=& \frac{\sqrt q}{2}\ketbraa{\Psi_-} + \left( 1-\frac{\sqrt q}{2} \right)\frac{\ketbraa{01}+\ketbraa{10}}{2}+O(\epsilon^2).
\end{eqnarray*}

The resulting state $\tilde{\rho}_G'$ violates the CHSH inequality up to $2\sqrt{1+q/4}$ (for $\epsilon\rightarrow0$) according to the Horodecki criterion. Hence, sequential measurements are necessary to reveal the nonlocality of $\rho_G'$, which therefore exhibits genuine hidden nonlocality.

\begin{figure}[h!]
  \centering
  \includegraphics[width=18cm]{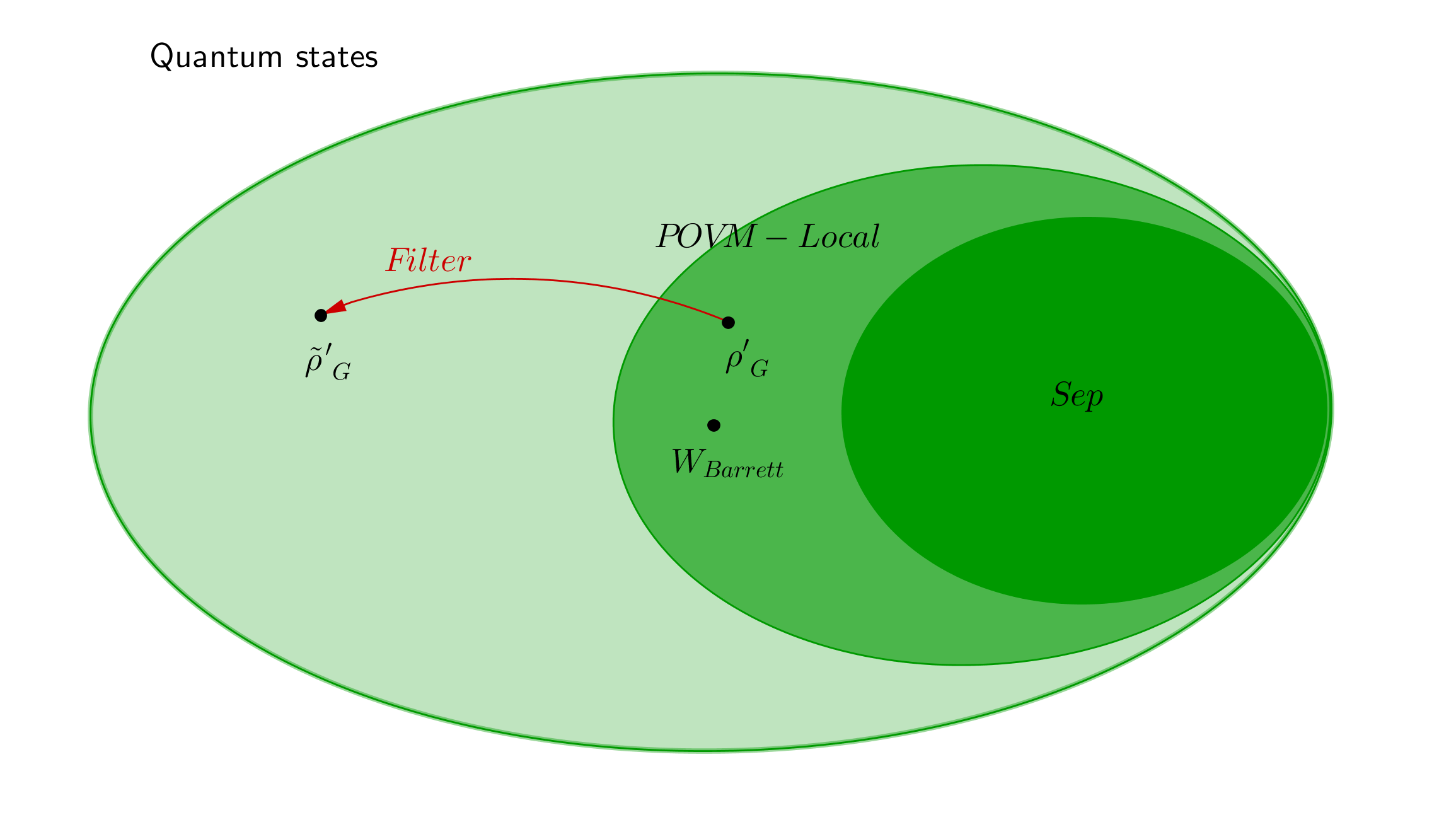}\\
  \caption{\small{The successful application of a local filtering (represented by the arrow in red) on the POVM-local state $\rho_G'$ results in a state $\tilde{\rho}_G'$ which violates CHSH, and thus does not belong to the set of local states.}}
\end{figure}

\section{Genuine and maximal}

Until now, throughout this chapter we have dealt with quantum states of $\mathbb{C}^2\otimes \mathbb{C}^2$. We conclude the chapter by saying that in \cite{HQBB} the authors also present an even more extremal case of this phenomenon, occurring with a state of $\mathbb{C}^3\otimes\mathbb{C}^2$, a qutrit-qubit state. Using the presented machinery, they showed that a state with genuine hidden nonlocality can violate maximally the CHSH inequality.

The state
\begin{equation*}
  \rho_E=q\ketbraa{\Psi_-}+(1-q)\ketbraa2\otimes\frac{I_2}{2}
\end{equation*}
can be shown to be local for dichotomic projective measurements (a protocol similar to Protocol \ref{proto geneva} should be performed). Then, applying Protocol \ref{proto geneva 2} to $\rho_E$ and taking $\sigma_A=\sigma_B=\ketbraa2$ we can simulate POVMs on the state
\begin{equation}\label{ayoria}
  \tilde{\rho_E}=\frac{1}{9}\left[ q\ketbraa{\Psi_-}+(3-q)\ketbraa2\otimes\frac{I_2}{2} +2q\frac{I_2}{2}\ketbraa2 + (6-2q)\ketbraa{22} \right],
\end{equation}
constructed via Eq. (\ref{seyia}).

To reveal the nonlocality of the above state, we apply filters of the form $F_A=F_B=\ketbraa0+\ketbraa1$. After successful filtering, the pure singlet state $\ketbraa{\Psi_-}$ is obtained, which we know to violate maximally CHSH. Hence, state (\ref{ayoria}) has genuine and maximal hidden nonlocality.

\addcontentsline{toc}{chapter}{Conclusion}

\chapter*{Conclusion}

In this master's thesis, our focus was to present some examples of local models and hidden nonlocality, providing the calculations underneath them and emphasizing the motivations as much as possible. In particular, the most seminal work on the subject, Werner's local model, was originally presented in a very knotty way. Perhaps a more detailed study of it makes easier to introduce new researchers to the topic; this work maybe is a first step towards that direction.

Naturally, several works about the subject do not appear here, and others are just shortly cited. Specifically in the case of local models, this is a bit frustrating, since the universe of local models is not very large. Barrett's model \cite{Bar}, for example, deserves more space, given its importance. Another important work is Ref. \cite{Mafalda}, in which it is studied the robustness of nonlocality to noise and it is presented a local model for isotropic states, which inspired Hirsch \emph{et al.} \cite{HQBB}. For the tripartite case, Ref. \cite{tripartite} presents a very interesting local model for projective measurements. A very complete and up-to-date review on the whole local models topic can be found in Ref. \cite{review}. About activation of nonlocality, we left a whole branch untouched, where are considered multiple copies of a state and quantum networks, in order to culminate in a Bell inequality violation \cite{BCPSW,Mur}.

What becomes clear is that this is a research area with a lot of potential, where representative examples are welcome and, with rare exceptions, general results still are only conjectured. For instance, is it possible to create a local model with hidden variables other than unit vectors? Is there any projective-local state which cannot simulate a POVM, meaning that POVMs do offer advantage in order to detect the nonlocality of quantum states? Is it possible to construct a local model for POVMs, in the multipartite case? Is there a quantum entangled state completely local, that is, that does not violate any Bell inequality, even in these more general scenarios of sequences of measurements and multiple copies? Or are entanglement and nonlocality, after all, one and the same thing, in this broader sense?\footnote{Related to this last question, there is the Peres Conjecture, which claims that states with bound entanglement are completely local. Recently, the conjecture was disproved \cite{peres}.} There is work to be done.



\end{document}